\documentclass[11pt]{article}
\usepackage{amsmath,amssymb,amsfonts,amsthm,epsfig}
\usepackage[usenames,dvipsnames]{xcolor}

\usepackage{bm,xspace}

\usepackage{cancel}
\usepackage{fullpage}
\usepackage{liyang}
\usepackage{framed}
\usepackage{verbatim}
\usepackage{enumitem}
\usepackage{array}
\usepackage{multirow}
\usepackage{afterpage}
\usepackage{mathrsfs}  

\usepackage{dsfont} 
\usepackage[normalem]{ulem}

\newcommand{\pparagraph}[1]{\medskip \noindent {\bf {#1}}}

\usepackage{caption} 

\usepackage{todonotes}

\makeatletter
\newtheorem*{rep@theorem}{\rep@title}
\newcommand{\newreptheorem}[2]{
\newenvironment{rep#1}[1]{
 \def\rep@title{#2 \ref{##1}}
 \begin{rep@theorem}\itshape}
 {\end{rep@theorem}}}
\makeatother
\theoremstyle{plain}

\newcommand{\ignore}[1]{}



\def\colorful{0}

\ifnum\colorful=1
\newcommand{\violet}[1]{{\color{violet}{#1}}}

\newcommand{\blue}[1]{{{\color{blue}#1}}}

\fi
\ifnum\colorful=0
\newcommand{\violet}[1]{{{#1}}}

\newcommand{\blue}[1]{{{#1}}}

\fi

\usepackage{boxedminipage}

\newreptheorem{theorem}{Theorem}
\newtheorem*{theorem*}{Theorem}
\newreptheorem{lemma}{Lemma}
\newreptheorem{proposition}{Proposition}
\newtheorem*{noclaim*}{Claim}

\newcommand{\uhr}{\upharpoonright}

\newcommand{\err}{\mathrm{err}}
\newcommand{\normal}{{\mathcal{N}(0,1)}}

\newcommand{\Gen}{\mathsf{Gen}}

\newcommand{\strip}{\mathrm{Strip}}
\newcommand{\inner}{\mathrm{Inner}}
\newcommand{\outter}{\mathrm{Outer}}

\renewcommand{\N}{\mathds{N}}

\renewcommand{\R}{\mathds{R}}

\newcommand{\unif}{\mathrm{unif}}
\newcommand{\pseudo}{\mathrm{pseudo}}

\newcommand{\hash}{\mathrm{hash}} 

\newcommand{\bucket}{\mathrm{bucket}}

\newcommand{\CNFLTF}{\textsc{CnfLtf}}

\newcommand{\CNF}{\mathrm{CNF}} 

\renewcommand{\wt}{\widetilde}

\begin{document}

\title{Fooling intersections of low-weight halfspaces}

\author{Rocco A.~Servedio\thanks{Supported by NSF grants CCF-1420349 and CCF-1563155. Email: {\tt rocco@cs.columbia.edu}} \\ Columbia University \and Li-Yang Tan\thanks{Supported by NSF grant CCF-1563122. Email: {\tt liyang@cs.columbia.edu}} \\ Toyota Technological Institute}

\maketitle

\begin{abstract}

A \emph{weight-$t$ halfspace} is a Boolean function $f(x)=\sign(w_1 x_1 + \cdots + w_n x_n - \theta)$ where each $w_i$ is an integer in $\{-t,\dots,t\}.$  We give an explicit pseudorandom generator that $\delta$-fools any intersection of $k$ weight-$t$ halfspaces with seed length $\poly(\log n, \log k,t,1/\delta)$. In particular, our result gives an explicit PRG that fools any intersection of any quasi$\poly(n)$ number of halfspaces of any $\polylog(n)$ weight to any $1/\polylog(n)$ accuracy using seed length $\polylog(n).$ 
Prior to this work no explicit PRG with non-trivial seed length was known even for fooling intersections of $n$ weight-1 halfspaces to constant accuracy.

The analysis of our PRG fuses techniques from two different lines of work on unconditional pseudorandomness for different kinds of Boolean functions.  We extend the approach of Harsha, Klivans and Meka \cite{HKM12} for fooling intersections of regular halfspaces, and combine this approach with results of Bazzi \cite{Bazzi:07} and Razborov \cite{Razborov:09} on bounded independence fooling CNF formulas.  Our analysis introduces new coupling-based ingredients into the standard Lindeberg method for establishing quantitative central limit theorems and associated pseudorandomness results.
\end{abstract}

 \thispagestyle{empty}

\newpage

\setcounter{page}{1}

\section{Introduction}

A \emph{halfspace}, or \emph{linear threshold function} (henceforth abbreviated LTF), over $\bits^n$ is a Boolean function $f$ that can be expressed as $f(x) = \sign(w_1 x_1 + \cdots + w_n x_n - \theta)$ for some real values $w_1,\dots,w_n,\theta$.  
LTFs are a natural class of Boolean functions which play a central role in many areas such as machine learning and voting theory, and have been intensively studied in complexity theory from many perspectives such as circuit complexity \cite{GHR:92,Razborov:92,Hastad:94,SO03}, communication complexity \cite{Nisan93thecommunication,Viola15}, Boolean function analysis \cite{Chow:61,GL:94,Peres:04,Servedio:07cc,ODBook}, property testing \cite{MORS:09random,MORS:10}, pseudorandomness \cite{DGJ+10:bifh,MZ13prg,GKM15} and more.

Because of the limited expressiveness of a single LTF (even a parity function over two variables cannot be expressed as an LTF), it is natural to consider Boolean functions that are obtained by combining LTFs in various ways.  
Perhaps the simplest and most natural functions of this sort are  \emph{intersections of LTFs}, i.e.~Boolean functions of the form $F_1 \wedge \cdots \wedge F_k$ where each $F_j$ is an LTF.  
Intersections of LTFs have been studied in many contexts including Boolean function analysis \cite{Kane14intersection,Sherstov13sicomp,Sherstov13combinatorica},  computational learning (both algorithms \cite{BlumKannan:97,KOS:04,KOS:08,Vempala:10} and hardness results \cite{KlivansSherstov:06,KhotSaket:11jcss}), and pseudorandomness \cite{GOWZ10,HKM12}.  We further note that the set of feasible solutions to an $\zo$-integer program with $k$ constraints corresponds precisely to the set of satisfying assignments of an intersection of $k$ LTFs; understanding the structure of these sets has been the subject of intensive study in computer science, optimization, and combinatorics.

This paper continues the study of intersections of LTFs from the perspective of unconditional pseudorandomness; in particular, we are interested in constructing explicit \emph{pseudorandom generators} (PRGs) for intersections of LTFs.   
Recall the following standard definitions:

\begin{definition}[Pseudorandom generator] \label{def:PRG}
A function $\Gen : \bits^r \to \bits^n$ is said to \emph{$\delta$-fool a function $F: \bits^n \to \{-1,1\}$ with seed length $r$} if
\[
\left|
\Ex_{\bU' \leftarrow \bits^r}\big[F(\Gen(\bU'))\big] -
\Ex_{\bU \leftarrow \bits^n}\big[F(\bU)\big] 
\right| \leq \delta.
\]
Such a function $\Gen$ is said to be a \emph{explicit pseudorandom generator that $\delta$-fools a class $\calF$ of $n$-variable functions} if $\Gen$ is computable by a deterministic uniform $\poly(n)$-time algorithm and $\Gen$ $\delta$-fools every function $F \in {\cal F}.$ \end{definition}

\subsection{Prior work} \label{sec:prior}

Before describing our results, we recall relevant prior work on fooling LTFs and intersections of LTFs.

\pparagraph{Fooling a single LTF.}  In \cite{DGJ+10:bifh} Diakonikolas et al.~showed that any $\tilde{O}(1/\delta^2)$-wise independent distribution over $\bits^n$ suffices to $\delta$-fool any LTF, and thereby gave a PRG for single LTFs with seed length $\tilde{O}(1/\delta^2) \cdot \log n$.  
Soon after, \cite{MZ13prg} gave a more efficient PRG for LTFs with seed length $O(\log n + \log^2(1/\delta)).$  
They did this by first developing an alternative  $\tilde{O}(1/\delta^2) \cdot \log n$ seed length PRG for \emph{regular} LTFs; these are LTFs in which no individual weight is large compared to the total size of all the weights (we give precise definitions later).  
\cite{MZ13prg} built on this PRG for regular LTFs  using structural results for LTFs and PRGs for read-once branching programs to obtain their improved $O(\log n + \log^2(1/\delta))$ seed length for fooling arbitrary LTFs. 
More recently, \cite{GKM15} gave a PRG which $\delta$-fools any LTF over $\bits^n$ using seed length $O(\log(n/\delta)(\log\log(n/\delta))^2)$; this is the current state-of-the-art for fooling a single LTF.

Since the approach of \cite{MZ13prg} for fooling regular LTFs is important for our discussion in later sections, we describe it briefly here.  
The \cite{MZ13prg} PRG for regular LTFs employs hashing and other techniques; its analysis crucially relies on the \emph{Berry--Ess\'een theorem} \cite{berry,esseen}.  
Recall that the Berry--Ess\'een theorem is an ``invariance principle'' for the distribution of linear forms; it (or rather, a special case of it) says that for $w$ a regular vector, the two random variables $w \cdot \bU$ and $w \cdot \bG$, where $\bU$ is uniform over $\bits^n$ and $\bG$ is drawn from the standard $n$-dimensional Gaussian distribution $\normal^n$, are close in CDF distance.  
Roughly speaking, the \cite{MZ13prg} PRG analysis for $\tau$-regular LTFs proceeds by showing that the limited randomness provided by their generator is sufficient to apply the Berry--Ess\'een theorem (over a certain set of roughly $1/\tau^2$ independent random variables).  
We give a more detailed description of the structure of the \cite{MZ13prg} PRG in Section~\ref{sec:approach}. 

\pparagraph{Fooling intersections of regular LTFs.}  Now we turn to results on fooling intersections of LTFs.
Essentially simultaneously with \cite{MZ13prg} (in terms of conference publication), \cite{HKM12} gave a PRG for intersections of \emph{regular} LTFs. 
Their PRG $\tilde{O}((\log k)^{8/5} \tau^{1/5})$-fools any intersection of $k$ many $\tau$-regular LTFs with seed length $O((\log n \log k)/\tau)$.  
As we discuss in in Section~\ref{sec:approach}, the \cite{HKM12} generator has the same structure as the \cite{MZ13prg} PRG for regular LTFs, but with different (larger) parameter settings and a significantly more involved analysis.  
At the heart of the correctness proof of the \cite{HKM12} PRG is a new invariance principle that \cite{HKM12} prove for \emph{$k$-tuples} $(w^{(1)} \cdot \bU, \dots, w^{(k)} \cdot \bU)$ of regular linear forms, generalizing the Berry--Ess\'een theorem which as described above applies to a single regular linear form.  
With this new invariance principle in hand, to prove their PRG theorem \cite{HKM12} argue (similar in spirit to \cite{MZ13prg}) that the limited randomness provided by their generator is sufficient for their new $k$-dimensional invariance principle.

Note that even the $k=1$ case of the invariance principle (the Berry--Ess\'een theorem) does not give a meaningful bound for non-regular linear forms.  
As a simple example, consider the trivial linear form $x_1$, which is highly non-regular: the two one-dimensional random variables $\bU_1$ and $\bG_1$, where $\bU_1$ is uniform over $\bits$ and $\bG_1$ is distributed according to $\normal$, have CDF distance $\approx 0.341$.  
And indeed the analysis of the \cite{HKM12} PRG only goes through for intersections of LTFs in which all the LTFs are regular. 
So while the \cite{HKM12} PRG has an extremely good (polylogarithmic) dependence on the number of LTFs in the intersection, the regularity requirement means that the \cite{HKM12} PRG theorem cannot be applied, for example, to fool the class of intersections of LTFs in which each weight is either 0 or 1.

\pparagraph{The PRG of Gopalan, O'Donnell, Wu, and Zuckerman.}  Around the same time, \cite{GOWZ10} gave a PRG that $\delta$-fools intersections of $k$ arbitrary LTFs with seed length $O((k \log(k/\delta) + \log n) \cdot \log(k/\delta))$, and indeed $\delta$-fools any depth-$k$ size-$s$ decision tree that queries LTFs at its internal nodes with seed length $\violet{O((k\log(k\violet{s}/\delta) + \log n)\cdot \log(k\violet{s}/\delta))}$. 
Their approach builds on the PRG of \cite{MZ13prg} for general LTFs; one central ingredient is a generalization of structural results for single LTFs used in \cite{MZ13prg} to $k$-tuples of LTFs.  
Both this generalization, and the read-once branching program based techniques from \cite{MZ13prg} (which are extended in \cite{GOWZ10} to the context of $k$-tuples of LTFs), necessitate a seed length which is at least linear in $k$.
So while the \cite{GOWZ10} PRG is is notable for being able to handle intersections of general LTFs, their seed length's linear dependence in $k$ means that their seed length is $n^{\Omega(1)}$ whenever $k=n^{\Omega(1)}$, and furthermore their result does not give a non-trivial PRG for intersections of $k \geq n$ many LTFs.

\subsubsection{A conceptual challenge}

We elaborate briefly on an issue related to the linear-in-$k$ dependence of the \cite{GOWZ10} generator discussed above.  A  standard approach to analyze non-regular LTFs, both in pseudorandomness and in other subfields of complexity theory such as analysis of Boolean functions and learning theory \cite{DS13,DRST14,DSTW14,DDS16,FGRW09,CSS16}, is to reduce the analysis of non-regular LTFs to that of regular LTFs via a ``critical index'' argument (see \cite{Servedio:07cc}).  
Indeed, most previous pseudorandomness results for classes involving non-regular LTFs and PTFs---general LTFs 
\cite{DGJ+10:bifh,MZ13prg}, functions of LTFs~\cite{GOWZ10}, degree-$d$ PTFs and functions of such PTFs \cite{DKNfocs10, MZ13prg,DDS14,DS14}---make use of such a reduction to the regular case. 
In working with functions that involve $k$ LTFs (or PTFs), this analysis (see \cite{DDS14,GOWZ10}) involves ``multi-critical-index'' arguments, originating in \cite{GOWZ10},  which necessitate an $\Omega(k)$ seed length dependence; indeed, this linear-in-$k$ dependence was highlighed in \cite{HKM12} as a conceptual challenge to overcome in extending their results to intersections of $k$ non-regular LTFs.  

In this work we give the first analysis that is able to handle an interesting class of functions involving $k$ non-regular LTFs while avoiding this linear-in-$k$ cost that is inherent to multi-critical-index based arguments, and in fact achieving a polylogarithmic dependence on $k$.

\subsection{Our main result:  fooling intersections of low-weight LTFs}

It is easy to see that every LTF $f: \bits^n \to \{-1,1\}$ has some representation as  $f(x)=\sign(w \cdot x - \theta)$ where the coefficients $w_1,\dots,w_n$ are all integers; a standard way of measuring the ``complexity'' of an LTF is by the size of its integer weights. 
It has been known at least since the 1960s \cite{MTT:61,Hong:87, Raghavan:88} that every $n$-variable LTF has an integer representation with $\max |w_i| \leq n^{O(n)}$,
 and H{\aa}stad has shown \cite{Hastad:94} that there are LTFs that in fact require $\max w_i = n^{\Omega(n)}$ for any integer representation.  
However, in many settings, LTFs with \emph{small integer weights} are of special interest.  Such LTFs are often the relevant ones in contexts such as voting systems or contexts where, e.g., biological or physical constraints may limit the size of the weights.  From a more theoretical perspective, it is well known that sample complexity bounds for many commonly used LTF learning methods, such as the Perceptron and Winnow algorithms, are essentially determined by the size of the integer weights.

We say that $f$ is a \emph{weight-$t$ LTF} if it can be represented as $f(x)=\sign(w \cdot x - \theta)$ where each $w_i$ is an integer satisfying $|w_i| \leq t.$  
Note that arguably the simplest and most natural LTFs --- unweighted threshold functions, with the majority function as a special case --- have weight 1.

Our main result is an efficient PRG for fooling intersections of low-weight LTFs:

\begin{theorem}[PRG for intersections of low-weight LTFs] 
\label{thm:prg-informal-low-weight}
For all values of $k,t\in \N$ and $\delta \in (0,1)$, there is an explicit pseudorandom generator that $\delta$-fools any intersection of $k$ weight-$t$ LTFs over $\bits^n$ with seed length $\poly(\log n,\log k, t, 1/\delta)$.
\end{theorem}

Recalling the results of \cite{HKM12,GOWZ10} described in Section~\ref{sec:prior}, prior to this work no explicit PRG with non-trivial seed length was known even for fooling intersections of $n$ weight-1 LTFs to constant accuracy.  (In fact, no 
$2^{0.99n}$-time algorithm was known for deterministic approximate counting of satisfying assignments of such an intersection; since such an algorithm is allowed to inspect the intersection of halfspaces which is its input, while a PRG is ``input-oblivious'', giving such an algorithm is an easier problem than constructing a PRG.) In contrast, our result gives an explicit PRG that fools any intersection of any quasi$\poly(n)$ number of LTFs of any $\polylog(n)$ weight to any $1/\polylog(n)$ accuracy using seed length $\polylog(n).$  For any $c>0$ our result also gives an explicit PRG with seed length $n^c$ that fools intersections of $\exp(n^{\Omega(1)})$ many LTFs of weight $n^{\Omega(1)}$ to accuracy $1/n^{\Omega(1)}.$
Recalling the correspondence between intersections of LTFs and $\zo$-integer programs, our PRG immediately yields new deterministic algorithms for approximately counting the number of feasible solutions to broad classes of $\{0,1\}$-integer programs.

\pparagraph{Our most general PRG result.}  We obtain Theorem~\ref{thm:prg-informal-low-weight} as an easy consequence of a PRG that fools a more general class of intersections of LTFs.  
To describe this class we require some terminology. 
We say that a vector $w$ over $\R^n$ is \emph{$s$-sparse} if at most $s$ coordinates among $w_1,\dots,w_n$ are nonzero.  
We similarly say that a linear threshold function $\sign(w \cdot x - \theta)$ is $s$-sparse if $w$ is $s$-sparse. 
Following \cite{HKM12}, we say that a linear form $w=(w_1,\dots,w_n)$ with norm $\|w\| := \left(\sum_{i=1}^n w_i^2\right)^{1/2}$ is \emph{$\tau$-regular} if $\sum_{i=1}^n w_i^4 \leq \tau^2 \|w\|^2$, and we say that a linear threshold function $\sign(w \cdot x - \theta)$ is $\tau$-regular if the linear form $w$ is $\tau$-regular.   
Finally, we say that $F: \bits^n \to \{-1,1\}$ is a \emph{$(k,s,\tau)$-intersection of LTFs} if $F=F_1 \wedge \cdots \wedge F_k$ where each $F_j$ is an LTF which is either $s$-sparse or $\tau$-regular.

Our most general PRG result is the following:

\begin{theorem}[Our most general PRG, informal statement] \label{thm:prg-informal}
For all values of $k,s\in \N$ and $\tau \in (0,1)$, there is an explicit pseudorandom generator with seed length $\poly(\log n,\log k, s,1/\tau)$ that fools any $(k,s,\tau)$-intersection of LTFs  to accuracy $\delta = \poly(\log k, \tau).$
\end{theorem}

In Section~\ref{sec:statements} we give the formal statement of Theorem~\ref{thm:prg-informal} and show how Theorem~\ref{thm:prg-informal-low-weight} follows from Theorem~\ref{thm:prg-informal}.

\section{Our approach} \label{sec:approach}

As explained in Section~\ref{sec:prior}, invariance-based arguments are not directly useful for our task of fooling intersections of low-weight LTFs, since the invariance principle does not give a non-trivial bound even for a single low-weight LTF.  
Nevertheless, we are able to show that a generator with the same structure as the \cite{MZ13prg,HKM12} generators (but now with slightly larger parameter settings than were used in the \cite{HKM12} generator) indeed fools any $(k,s,\tau)$-intersection of LTFs.  
We do this via an analysis that brings in ingredients that are novel in the context of fooling intersections of LTFs; in particular, we use results of Bazzi \cite{Bazzi:07} and Razborov \cite{Razborov:09} on bounded independence fooling depth-2 circuits.

How are depth-2 circuits relevant to intersections of LTFs?  
A starting point for our work is to re-express a $(k,s,\tau)$-intersection of LTFs using a different representation, in which we replace each $s$-sparse LTF by a CNF formula computing the same function over $\bits^n.$  
The following is an immediate consequence of the fact that any $s$-sparse LTF depends on at most $s$ variables:
\blue{\begin{fact} 
\label{fact:decompose} 
Let $F$ be a $(k,s,\tau)$-intersection of LTFs.  Then $F \equiv H \wedge G$, where 
\begin{itemize}
\item $H$ is the intersection of at most $k$ many $\tau$-regular LTFs. 
\item $G$ is a width-$s$ CNF formula with at most $k \cdot 2^s$ clauses;
\end{itemize} 
\end{fact} }

\noindent 
We refer to a function of the form \blue{$H \wedge G$} as above as a \emph{$(k,s,\tau)$-$\CNFLTF$}. 
We can thus restate our goal as that of designing a PRG to fool any $(k,s,\tau)$-$\CNFLTF$:  with this perspective it is not surprising that pseudorandomness tools for fooling CNF formulas can be of use.

\subsection{The structure of our PRG} \label{sec:structure}

To describe our approach we need to explain the general structure of the PRG which is used in \cite{MZ13prg} for regular LTFs, in \cite{HKM12} for intersections of regular LTFs, and in our work for $(k,s,\tau)$-intersections of LTFs.  
The construction  uses an $r_\hash$-wise independent family $\calH$ of hash functions $h : [n] \to [\ell]$, and an $r_\bucket$-wise independent generator outputting strings in $\bits^n$, which we denote $\mathscr{G}$.  
The overall generator, which we denote $\Gen$, on input $(h,X^{(1)},\dots,X^{(\ell)})$ outputs the string $\Gen(h,X^{(1)},\dots,X^{(\ell)}) := Y \in \bits^n,$ where $Y_{h^{-1}(b)} = \mathscr{G}(X^{(b)})_{h^{-1}(b)}$ for all $b \in [\ell]$. 
(Here and elsewhere, for $Y$ an $n$-bit string and $S \subseteq [n]$ we write $Y_S$ to denote the $|S|$-bit string obtained by restricting $Y$ to the coordinates in $S$.)

The \cite{MZ13prg} PRG for ${\tau}$-regular LTFs instantiates this construction with 
\[
\ell = 1/{\tau}^2,  \ \ {r_\hash} = 2, \ \ \text{and} \ \ {r_\bucket}=4,
\]
while the \cite{HKM12} PRG for intersections of $k$ many ${\tau}$-regular LTFs takes 
\[
\ell=1/{\tau}, \ \ {r_\hash} = 2 \log k, \ \ \text{and} \ \ {r_\bucket} =4\log k.
\]
\noindent We state the exact parameter settings which we use to fool $(k,s,\tau)$-intersections of LTFs in Section~\ref{sec:our-PRG} (the specific values are not important for our discussion in this section).

\subsection{Sketch of the \cite{MZ13prg,HKM12} analysis}

As our analysis (sketched in Section~\ref{sec:our-analysis}) builds on \cite{MZ13prg,HKM12}, in this subsection we sketch the \cite{MZ13prg,HKM12} arguments establishing correctness of the PRG $\Gen$ for regular LTFs and intersections of regular LTFs.

A high-level sketch of the \cite{MZ13prg} analysis showing that $\Gen$ fools any regular LTF $F(x)=\sign(w \cdot x - \theta)$ is as follows:  the hash function $h : [n]\to [\ell]$ partitions the $n$ coefficients $w_1,\dots,w_n$ into $\ell$ buckets.  
The pairwise independence of $\bh \leftarrow \calH$ and the regularity of $w$ are together used to show  that each of the $\ell$ buckets receives essentially the same amount of ``coefficient weight.''  
The idea then is to view the sum $w \cdot \bY$, where $\bY$ is the output of the generator, as a sum of $\ell$ \emph{independent} random variables (note that the inputs $\bX^{(1)},\dots,\bX^{(\ell)} \in \bits^r$ to $\Gen$ are indeed mutually independent), one for each bucket, and use the Berry--Ess\'een theorem on that sum.\footnote{Note that if the weight vector $w$ is non-regular, then it is in general impossible for any hash function, even a fully independent one, to spread the coefficient weight out evenly among the $\ell$ buckets, and consequently the Berry--Ess\'een theorem cannot be applied (as, intuitively, it requires that no individual random variable summand is ``too heavy'' compared to the ``total weight'' of the sum). 
This is why the overall approach requires regularity.} 
The four-wise independence of $\mathscr{G}$ is used to ensure that each of the $\ell$ summands---the $b$-th summand corresponding to $w_{\bh^{-1}(b)} \cdot \bY_{\bh^{-1}(b)}$, the contribution from the $b$-th bucket---has the moment properties that are required to apply the Berry--Ess\'een theorem.  
Note that in this analysis the Berry--Ess\'een theorem is used as a ``black box.''

Since \cite{HKM12} have to prove the $k$-dimensional invariance principle that they use in place of the Berry--Ess\'een theorem, their analysis is necessarily more involved, but at a high level it follows a similar approach to the \cite{MZ13prg} analysis sketched above.  
A sketch of their argument that $\Gen$ fools any intersection $F=F_1 \wedge \cdots \wedge F_k$ of regular LTFs is as follows:

\begin{enumerate}

\item \cite{HKM12} first argue that for \emph{any} smooth test function $\psi: \R^k \to [0,1]$---replacing the ``hard threshold'' function $\mathbf{1}(v_1 \leq \theta_1) \cdot \mathbf{1}(v_2 \leq \theta_2) \cdot \cdots  \mathbf{1}(v_k \leq \theta_k)$, which corresponds to $k$-dimensional CDF distance---the pseudorandom distribution output by the generator fools the test function $\psi$ relative to an $\normal^n$ Gaussian input to $\psi.$  
This is done by

\begin{enumerate}

\item first arguing (similar to \cite{MZ13prg}) that the $(2 \log k)$-wise independent hash function $\bh\leftarrow\calH$ and the regularity of each LTF $F_k$ together ``spread the coefficient weight'' of the $k$ LTFs roughly evenly among the $\ell$ buckets (we note that this part of the argument has nothing to do with the function $\psi$);

\item then a hybrid argument across the $\ell$ buckets, using the smoothness of $\psi$ and moment properties of the random variables corresponding to the $\ell$ buckets (which now follow from the $(4 \log k)$-wise independence of $\mathscr{G}$), is used to bound
\begin{equation} \left| \Ex_{\bY\leftarrow\Gen}\big [\psi(w^{(1)}\cdot \bY,\ldots,w^{(k)}\cdot \bY)\big] - \Ex_{\bG\leftarrow \normal^n}\big[\psi(w^{(1)}\cdot \bG,\ldots,w^{(k)}\cdot \bG)\big] \right|. 
\label{eq:smooth-small}
\end{equation}
\end{enumerate} (Such a hybrid argument is a central ingredient in the Lindeberg-style ``replacement method'' proof of the Berry--Ess\'een theorem, and is also used in \cite{HKM12}'s proof of their invariance principle for intersections of $k$ regular LTFs.) We note that multidimensional Taylor's theorem plays a crucial role in bounding the difference in expectation between $\psi$ applied to two random variables, which is done to ``bound the distance'' at each step of the hybrid.

\item Next \cite{HKM12} use a particular smooth function $\psi^\ast$ based on a result of Bentkus \cite{Bentkus:90} and a Gaussian surface area bound for intersections of $k$ halfspaces due to Nazarov \cite{Nazarov:03}  to pass from fooling the smooth test function $\psi^\ast$ to fooling the ``hard threshold'' function corresponding to CDF distance.  This essentially amounts to using the fact that (\ref{eq:smooth-small}) is small to show that $\left| \Ex_{\bY\leftarrow\Gen} [F(\bY)] - \Ex_{\bG\leftarrow \normal^n} [F(\bG)] \right|$
is also small.  
Given this, the fact that the generator fools $F$, i.e.~that $\left|\Ex_{\bY \leftarrow\Gen}[F(\bY)] - \Ex_{\bX \leftarrow\bits^n}[F(\bX)]\right|$ 
is small, follows from \cite{HKM12}'s invariance principle, which bounds 
$\left|\Ex_{\bG\leftarrow \normal^n}[F(\bG)] - 
\E_{\bX\leftarrow\bits^n}[F(\bX)]\right|$.
We note that this second step of \cite{HKM12}'s analysis does not use regularity of the $F_j$'s at all (but their invariance principle does require that each $F_j$ is regular).

\end{enumerate}

\subsection{Sketch of our analysis} \label{sec:our-analysis}

Here we give an overview of our proof that $\Gen$, with suitable parameters, fools any $(k,s,\tau)$-$\CNFLTF$ \blue{$F=H \wedge G$}. Recall that $H$ is an intersection of $k$ many $\tau$-regular LTFs and $G$ is a $(k \cdot 2^s)$-clause CNF, and that the  difference between our task and that of \cite{HKM12} is that we must handle the CNF $G$ in addition to the intersection of regular LTFs $H$.  While it is not difficult to see, as a consequence of \cite{Bazzi:07,Razborov:09}, that the \cite{HKM12} generator with suitable parameters (i) fools $H$, and (ii) fools $G$, it is far from clear \emph{a priori} that it fools $H \wedge G$.  We show this via a rather delicate argument, which involves a novel extension of the Lindeberg method that is at the heart of all PRGs in this line of work~\cite{GOWZ10,MZ13prg,HKM12}.  To surmount the technical challenges that arise in our setting (which we described next),  our analysis features several new ingredients which are not present in the analyses of~\cite{GOWZ10,MZ13prg,HKM12}, or indeed in other Lindeberg-type proofs of quantitative central limit theorems that we are aware of. The ideas in this new style of coupling-based analysis, which we outline in Section~\ref{sec:new-ingredients} below, may be of use elsewhere.

\paragraph{The standard Lindeberg setup, and a new challenge in our setting.} As is standard in Lindeberg-style proofs, our analysis focuses on a particular smooth test function, which \violet{for us takes $k+1$ arguments and which} we denote $\psi^\ast_{k+1}$. This should be thought of as the $(k+1)$-variable version of the smooth function of Bentkus \cite{Bentkus:90}, which was used by \cite{HKM12} as mentioned in the preceding subsection.  Crucually, while $\psi^\ast_{k+1}$ maps all of $\R^{k+1}$ to $[-1,1]$, in our arguments this test function will only ever receive inputs 
from \blue{$\R^k \times \{\pm 1\}$}; indeed, its \blue{last ($(k+1)$-st) coordinate} will always be a Boolean value which is the output of the CNF $G$.  

The heart of our proof lies in showing that for this specific smooth test function $\psi^\ast_{k+1}$ (which should be thought of as a proxy for \blue{$\text{{\sc And}}(\sign(\violet{v_1} - \theta),\dots,\sign(\violet{v_{k}} - \theta_{k}), v_{k+1}))$}), the pseudorandom distribution output by the generator fools the test function $\psi^\ast_{k+1}$ relative to a uniform random input drawn from $\bits^n.$  This is done by means of a hybrid argument, the analysis of which (like that of \cite{GOWZ10,HKM12}) employs a multidimensional version of Taylor's theorem.  However, the fact that the distinguished \blue{last} coordinate of $\psi^\ast_{k+1}$ always receives a $\{\pm 1\}$-valued input---in particular, an input whose magnitude changes by a large amount (namely 2) when it does change---introduces significant challenges in using the multidimensional Taylor's theorem.   Recall that Taylor's theorem quantifies the following intuition: roughly speaking, if the input to a smooth function $\psi$ is only changed by a small amount $\Delta$, then the resulting change in its output value, $\psi(v + \Delta)- \psi(v)$, is correspondingly small as well.  Naturally, if $\Delta$ is large then Taylor's theorem does not give useful bounds.  

\subsubsection{New ingredients in our approach}
\label{sec:new-ingredients} 

Taylor's theorem is the core ingredient in Lindeberg-style proofs of invariance principles (see e.g.~\cite{Tao254A} and Chapter~11 of \cite{ODBook}) and associated pseudorandomness results (see e.g.~\cite{GOWZ10,MZ13prg,HKM12}), where it is used to bound the distance incurred by a single step of the hybrid argument. As mentioned above, in order for Taylor's theorem to give a useful bound when it is applied to re-express $\psi^\ast_{k+1}(v+\Delta)$ (in terms of $\psi^\ast_{k+1}(v)$, various derivatives of $\psi^\ast_{k+1}$ at $v$, $\Delta$, and an error term), the quantity $\Delta$ must be ``small.''  This is a problem in our context since the distinguished \blue{last} coordinate of $\psi^\ast_{k+1}$'s argument (the output of the CNF $G$) is $\{\pm 1\}$-valued, so the \blue{last} coordinate of $\Delta$ alone may already be as large as 2.  We get around this difficulty by utilizing a carefully chosen coupling between two adjacent hybrid random variables and decomposing each of the two relevant arguments \violet{to which $\psi^\ast_{k+1}$ is applied}  (each of which is a random variable) in a very careful way.  One of these random variables is expressed as $\bv + \bDelta^\unif$ (corresponding to ``filling in the current bucket uniformly at random'') and the other is $\bv + \bDelta^\pseudo$ (corresponding to ``filling in the current bucket pseudorandomly"); roughly speaking, in order to succeed our analysis  must show that the magnitude of $\E[\psi^\ast_{k+1}(\bv + \bDelta^\unif)] - \E[\psi^\ast_{k+1}(\bv + \bDelta^\pseudo)]$ is suitably small.   The key property of the coupling we employ is that it 
ensures that the \blue{last} coordinates of both random variables $\bDelta^\unif$ and $\bDelta^\pseudo$ are almost always zero; in fact, one of them will actually be always zero, see Equation (\ref{eq:always-zero}).  (We note that if no coupling is used then the \blue{last} coordinate of $\bDelta^\pseudo$ can be as large as 2 with constant probability.) The existence of such a favorable coupling follows from the fact that each bucket of $\Gen$ is, by virtue of its bounded independence and the results of Bazzi \cite{Bazzi:07} and Razborov \cite{Razborov:09}, ``sufficiently pseudorandom'' to fool CNF formulas.

However, the way that we structure the random variables $\bv, \bDelta^\unif$, and $\bDelta^\pseudo$ to ensure that the \blue{last} coordinate of each $\bDelta$ is almost always small (as discussed above), introduces a new complication, which is that now the random variables $\bv$ and $\bDelta^\unif$ are not independent (and neither are $\bv$ and $\bDelta^\pseudo$). This situation does not arise in standard uses of the Lindeberg method, either in proving invariance principles or in applications to pseudorandom generators. In all of these previous proofs, independence is used to show that various first derivative, second derivative, etc.~terms in the Taylor expansions for the two adjacent random variables cancel out perfectly upon subtraction (using matching moments).  To surmount this lack of independence, we exploit the fact that our coupling lets us re-express the coupled joint distribution (over a pair of vectors in \blue{$\R^{k} \times \{ \pm 1\}$}) as a mixture of three joint distributions over pairs of $(k+1)$-dimensional vectors in such a way that one component of the mixture is entirely supported on \blue{$(\R^k \times \{1\}) \times (\R^k \times \{1\}),$} one is entirely supported on 
\blue{$(\R^k \times \{-1\}) \times (\R^k\times \{-1\}),$} and the third has a very small mixing weight.  Under each of the first two joint distributions (supported entirely on pairs that agree in the \blue{last} coordinate), $\bv$ and $\bDelta^\unif$ will indeed be independent, and so will $\bv$ and $\bDelta^\pseudo$.  

However, performing the hybrid method using these conditional distributions presents another challenge: while now $\bv$ and $\bDelta^\unif$ are independent (and likewise for $\bv$ and $\bDelta^\pseudo$), the moments of these conditional random variables may not match perfectly.   We deal with this by exploiting the fact that each pseudorandom distribution that we consider ``filling in a single bucket'' can in fact fool, to very high accuracy, any of $\poly(n)$ many new circuits which arise in our analysis of the multidimensional Taylor expansion (intuitively, these are ``slightly augmented'' CNFs or DNFs).  This allows us to show that while we do not get perfect cancellation, the \violet{relevant} moments under the conditional distributions are adequately close to each other.   Finally, our coupling-based perspective also allows us to bound the (crucial) final error term resulting from Taylor's theorem by reducing its analysis to that of the corresponding error term in \cite{HKM12}.

The above is a sketch of how we show that $\Gen$ fools the smooth test function $\psi^\ast_{k+1}$.  To pass from fooling $\psi^\ast_{k+1}$ to fooling the ``hard threshold'' {\sc And} function, we combine the \cite{HKM12} invariance principle with a simple relationship, Claim~\ref{claim:psi-extra-coordinate}, which we establish between the anti-concentration of the $(k+1)$-dimensional input to the $\psi^\ast_{k+1}$ function (with its distinguished \blue{last} coordinate corresponding to outputs of the CNF) and its $k$-dimensional marginal which excludes the \blue{last} coordinate (all coordinates of which correspond to outputs of regular linear forms, i.e.~the setting of \cite{HKM12}).

\section{Notation and preliminaries} \label{sec:prelims}

\pparagraph{LTFs and regularity.} We recall that a linear threshold function (LTF) is a function of the form $\sign(w \cdot x - \theta)$, where $\sign(z)$ is 1 if $z > 0$ and is $-1$ otherwise.  We view $-1$ as {\sc True} and $1$ as {\sc False} throughout the paper.

We write $W \in \R^{n\times k}$ to denote the matrix whose $j$-th column is the weight vector of the $j$-th LTF in an intersection of $k$ LTFs.  We assume that each such LTF has been normalized so that its weight vector has norm 1. For $j \in [k]$ (indexing one of the LTFs) we write $W^j$ to denote the $j$-th column of $W$ (so $\|W^j\|=1$ for all $j$), and for $B \sse [n]$ (a subset of variables) we write $W_B$ to denote the matrix formed by the rows of $W$ with indices in $B$.  Combining these notations, $W^j_B$ denotes the $|B|$-element column vector which is obtained from $W^{j}$ by taking those entries given by the indices in $B$. Throughout the paper we will write $\vec{\theta}$ to denote the $k$-tuple $\vec{\theta} = (\theta_1,\ldots,\theta_k) \in \R^k$. 

We say that a vector $w \in \R^n$ is \emph{$\tau$-regular} if $\sum_{i=1}^n w_i^4 \leq \tau^2 \|w\|^2$, and that it is \emph{$s$-sparse} if it has at most $s$ non-zero entries.  We use the same terminology to refer to an LTF $\sign(w \cdot x - \theta).$  We say that a matrix $W \in \R^{n \times k}$ is $\tau$-regular if each of its columns is $\tau$-regular.

A \emph{restriction} $\rho$ fixing a subset $S \subseteq [n]$ of $n$ input variables is an element of $\{0,1\}^S$; it corresponds to setting the variables in $S$ in the obvious way and leaving the variables outside $S$ free.  Given an $n$-variable function $f$ and a restriction $\rho$ we write $f \uhr \rho$ to denote the function obtained by setting some of the input variables as dictated by $\rho$.

\pparagraph{Probability background.}
We recall some standard definitions of bounded-independence distributions and hash families.
A distribution $\calD$ over $\bits^n$ is \emph{$r$-wise independent} if for every $1 \leq i_1 < \cdots < i_r \leq n$ and every $(b_1,\dots,b_r) \in \bits^r$, we have
\[
\Prx_{\bX \leftarrow \calD}\big[\bX_{i_1} = b_1 \text{~and~} \cdots \text{~and~}\bX_{i_r} = b_r\big] = 2^{-r}.
\]
We recall the results of \cite{Bazzi:07,Razborov:09} which state that bounded-independence distributions fool CNF formulas:

\begin{theorem} [Bounded independence fools \violet{depth-2 circuits}] \label{thm:BR}
Let $f$ be any $M$-clause CNF formula or $M$-term \violet{DNF} formula.  Then $f$ is $\delta$-fooled by any $O((\log(M/\delta))^2)$-wise independent distribution.
\end{theorem}

A family $\calH$ of functions from $[n]$ to $[\ell]$ is said to be an \emph{$r$-wise independent hash family} if for every $1 \leq i_1 < \cdots < i_r \leq n$ and $(j_1,\dots,j_r) \in [\ell]^r$, we have
\[
\Prx_{\bh \leftarrow \calH}\big[\bh(i_1) = j_1 \text{~and~} \cdots \text{~and~} \bh(i_r)=j_r\big]=\ell^{-r}.
\]
When $S$ is a set the notations $\Prx_{\bX \leftarrow S}[\cdot], \Ex_{\bX \leftarrow S}[\cdot]$ indicate that the relevant probability or expectation is over a uniform draw of $\bX$ from set $S$.  Throughout the paper we use bold fonts such as $\bX,\bU,\bh,$ etc. to indicate random variables.  We write $\normal$ to denote the standard normal distribution with mean 0 and variance 1.

\pparagraph{Calculus.} We say that a function $\psi: \R^k \to \R$ is \emph{smooth} if its first through fourth derivatives are uniformly bounded.  For smooth $\psi: \R^k \to \R$, $v \in \R^k$, and $j_1,\dots,j_r \in [k]$, we write $(\partial_{j_1,\dots,j_r} \psi)(x)$ to denote $\partial_{j_1} \partial_{j_2} \cdots \partial_{j_r} \psi(x)$, and for $s=1,2,\dots$ we write
\[
\|\psi^{(s)}\|_1 \quad \quad \text{to denote} \quad \quad
\sup_{v\in \R^k} \left\{ 
\sum_{j_1,\dots,j_s \in [k]} |(\partial_{j_1,\dots,j_s}\psi)(v)|
\right\}.
\]

Given indices $j_1,\dots,j_r \in [k]$, we write $(j_1,\dots,j_r)!$ to denote $s_1! s_2! \cdots s_k!$, where for each $\ell \in [k]$, $s_\ell$ denotes the number of occurrences of $\ell$ in $(j_1,\dots,j_r).$  We will use the following form of multidimensional Taylor's theorem (see e.g.~Fact~4.3 of \cite{HKM12}):

\begin{fact}[Multidimensional Taylor's theorem] \label{fact:taylor}
Let $\psi: \R^k \to \R$ be smooth and let $v,\Delta \in \R^k.$  Then
\begin{align*}
\psi(v+\Delta) &= 
\psi(v) 
+ \sum_{j \in [k]} (\partial_j \psi)(v) \Delta_j
+ \sum_{j,j' \in [k]} {\frac 1 {(j,j')!}} (\partial_{j,j'} \psi)(v) \Delta_{j} \Delta_{j'} \\
& + \sum_{j,j',j'' \in [k]} {\frac 1 {(j,j',j'')!}} (\partial_{j,j',j''} \psi)(v) \Delta_{j} \Delta_{j'} \Delta_{j''} 
+ \err(v,\Delta),
\end{align*}
where $|\err(v,\Delta)| \leq \|\psi^{(4)}\|_1 \cdot \max_{j \in [k]} |\Delta_j|^4.$
\end{fact}

\pparagraph{Useful results from  \cite{HKM12}.}  The following notation will be useful:   for $0<\lambda<1$, $k \geq 1$, and $\vec{\theta} = (\theta_1,\ldots,\theta_k) \in \R^k$, we define
\[
\inner_{k,\vec{\theta}} = \big\{v \in \R^k\colon v_j \le \theta_j \text{\ for all $j\in [k]$}\big\}, \quad 
\outter_{\lambda,k,\vec{\theta}} = \big\{v \in \R^k \colon v_j \ge \theta_j + \lambda \text{~for some $j\in [k]$}\big\},
\]
\[
\strip_{\lambda,k,\vec{\theta}} = \R^k \setminus (\inner_{k,\vec{\theta}} \cup \outter_{\lambda,k,\vec{\theta}}).
\]

We recall the main result of \cite{HKM12}:

\begin{theorem} [Invariance principle for polytopes, Theorem~3.1 of \cite{HKM12}] \label{thm:HKM-invariance}
Let $W \in \R^{n \times k}$ be $\tau$-regular with each column $W^j$ satisfying $\|W^j\|=1$.  Then for all $\vec{\theta} \in \R^k$, we have
\[
\left|
\Prx_{\bU \leftarrow \bits^{n}}\big[\,W^T \bU \in \inner_{k,\vec{\theta}}\,\big]
-
\Prx_{\bG \leftarrow \normal^n}\big[\,W^T \bG \in \inner_{k,\vec{\theta}}\,\big]
\right| 
= O\big((\log k)^{8/5} (\tau \log(1/\tau))^{1/5}\big).
\]
\end{theorem}

We will also use the following anti-concentration bound for Gaussian random variables (which is an easy consequence of the $O(\sqrt{\log k})$ Gaussian surface area upper bound of Nazarov \cite{Nazarov:03} for intersections of $k$ LTFs):

\begin{theorem} [Anti-concentration bound for Gaussian random variables landing in a strip, Lemma~3.4 of \cite{HKM12}]
\label{thm:Gaussian-anticoncentration}
For all $\vec{\theta} \in \R^k$ and all $0 < \lambda < 1$, we have
\[
\Prx_{\bG \leftarrow \normal^n}\big[\,W^T \bG \in \strip_{\lambda,k,\vec{\theta}}\,\big] = O(\lambda \sqrt{\log k}).
\]
\end{theorem}

\section{Our PRG and the statements of our main results} \label{sec:our-PRG}

Our PRG for $(k,s,\tau)$-intersections of LTFs is the generator $\Gen$ described in Section~\ref{sec:structure}, instantiated with the following parameters: 
\begin{align*}
\ell &=  1/\tau, \\
r_\hash &= 2\log k, \\
r_\bucket &= 4 \log k +  O((\log (M/\delta_\CNF))^2 
\end{align*}
where 
\[ M = k\cdot 2^s \quad \text{and} \quad \delta_\CNF = 1/\poly(n) \] 
(the exact value for $\delta_\CNF$ will be specified later).  By standard constructions of $r_\hash$-wise independent hash families and $r_\bucket$-wise independent random variables, the total seed length of our generator is  
\begin{align*}
O( \log(n\log \ell) \cdot r_\hash + \ell \cdot (\log n)\cdot r_\bucket) &= O\left({\frac 1 \tau} \cdot \log n \cdot \left(\log k + s + \log n\right)^2\right)\\
& = \poly(\log n, \log k, s, 1/\tau).
\end{align*}

\subsection{Formal statements of our main results} 
\label{sec:statements} 

We begin with our most general PRG result: 

\begin{reptheorem}{thm:prg-informal}
For all values of $k,s\in \N$ and $\tau \in (0,1)$, the pseudorandom generator $\Gen$ instantiated with the parameters above fools the class of $(k,s,\tau)$-intersections of LTFs to accuracy 
 \begin{equation}  \delta := O((\log k)^{8/5} (\tau \log(1/\tau))^{1/5})) \label{eq:delta-tau}
\end{equation} 
with seed length
$\poly(\log n,\log k, s, 1/\tau)$.  
\end{reptheorem}

Our PRG for the intersections of low-weight LTFs (Theorem~\ref{thm:prg-informal-low-weight}) follows as a consequence of Theorem~\ref{thm:prg-informal} via the following observation:

\begin{observation}[Sparse-or-regular dichotomy] \label{obs:sparse-to-reg}
Let $F(x)=\sign(w \cdot x - \theta)$ be a weight-$t$ LTF.   Then for any $s$, either $F$ is $s$-sparse or $F$ is $({{t}/{\sqrt{s+1}}})$-regular.
\end{observation}
\begin{proof}
Suppose that $F$ is not $s$-sparse; for notational convenience we may suppose that $w=(w_1,\dots,w_{s'},0,\dots,0)$ where $s' \geq s+1$ and for $1 \leq i \leq s'$ each $w_i$ is a nonzero integer in $\{-t,\ldots,t\}.$  
Normalize the weights by setting $u_i = {w_i} /{\|w\|}$ for $i=1,\dots,n$.  
We have $F(x) = \sign(u \cdot x - \theta/\|w\|)$ where $\|u\|=1$.  

To show that $F$ is $({{t}/{\sqrt{s+1}}})$-regular we must show that $\sum_{i=1}^{s'} u_i^4 \leq {{t^2}/({s+1})}.$
We have
\[
\sum_{i=1}^n u_i^4 \leq \left(\max_{1 \leq j \leq s'} u_j^2\right) \cdot \sum_{i=1}^n u_i^2 = 
\max_{1 \leq j \leq s'} u_j^2 \leq {\frac {t^2}{\|w\|^2}} \leq {\frac {t^2}{s+1}},
\]
where the last inequality holds because ${ {(s+1)}/{\|w\|^2}} \leq { {s'}/{\|w\|^2}} \leq
\sum_{i=1}^{s'} u_i^2 = 1.$ 
\end{proof}

\begin{reptheorem}{thm:prg-informal-low-weight}
For all $k,t\in \N$ and $\delta \in (0,1)$, there is an explicit pseudorandom generator with seed length $\poly(\log n, \log k, t, 1/\delta)$ that $\delta$-fools any intersection of $k$ weight-$t$ LTFs.
\end{reptheorem}

\begin{proof}[Proof of Theorem~\ref{thm:prg-informal-low-weight} assuming Theorem~\ref{thm:prg-informal}]
We fix 
\[ \tau := \tilde{\Theta}\left(\frac{\delta^{5}}{(\log k)^{8}}\right) \] 
so as to satisfy (\ref{eq:delta-tau}).  By Observation~\ref{obs:sparse-to-reg}, we have that every weight-$t$ LTF is either $\tau$-regular or $(s := (t/\tau)^2)$-sparse.  By our choice of $\tau$, the parameters $\ell, r_\hash$, and $r_\bucket$ of the pseudorandom generator $\Gen$ instantiated with our parameters are all bounded by $\poly(\log n,\log k, t, 1/\delta)$, and hence the overall seed length is indeed
\[ O\left(\log(n\log \ell) \cdot r_\hash + \ell \cdot (\log n)\cdot r_\bucket \right)= \poly(\log n,\log k, t, 1/\delta)
\] 
 as claimed. 
\end{proof}

The remainder of this paper will be devoted to proving Theorem~\ref{thm:prg-informal}. 

\section{Fooling the smooth test function $\psi^\ast_{k+1}$} 

An intermediate goal, which in fact takes us most of the way to establishing Theorem~\ref{thm:prg-informal}, is to
show that $\Gen$ fools a particular smooth test function $\psi^\ast_{\lambda,k+1,\blue{(\vec{\theta},0)}}$.  In this section we define this smooth test function, establish some of its basic properties, and formally state our intermediate goal (Theorem~\ref{thm:our-smooth-prg}  below).

\subsection{The smooth test function $\psi^\ast_{\lambda,k+1,\blue{(\vec{\theta},0)}}$ and its basic properties}

As discussed in Section~\ref{sec:our-analysis}, our analysis crucially features a particular smooth function $\psi^\ast_{\lambda,k+1,\blue{(\vec{\theta},0)}}  : \R^{k+1} \to [-1,1]$, which is essentially the $(k+1)$-dimensional version of a function due to Bentkus \cite{Bentkus:90}.  Fact~\ref{fact:cor-3.6} below states the key properties of this function.

\begin{fact} [Main result of \cite{Bentkus:90}, see Theorem~3.5 of \cite{HKM12}] \label{fact:cor-3.6}
For all positive integers $k$, $0 < \lambda < 1$, and $\vec{\theta}\in \R^k$, there exists a smooth function $\psi^\ast_{\lambda,k,\vec{\theta}}: \R^{k} \to[-1,1]$ such that the following holds:  for every $s=1,2,\dots$, we have
$\|(\psi^\ast_{\lambda,k,\vec{\theta}})^{(s)}\|_1 \leq C \log^{s-1}(k+1)/\lambda^s$, and for all $v \in \R^{k}$, we have
\begin{equation} \label{eq:psi-properties}
\psi^\ast_{\lambda,k,\vec{\theta}}(v) = \begin{cases}
-1 & \text{if~} v \in \inner_{k,\vec{\theta}}\\
1  & \text{if~} v \in \outter_{\lambda,k,\vec{\theta}}\\
\in [-1,1] & \text{otherwise (i.e. if~} v \in \strip_{\lambda,k,\vec{\theta}}).
\end{cases}
\end{equation}
\end{fact}
For intuition, the test function $\psi^\ast_{\lambda,k,\vec{0}}: \R^{k} \to [-1,1]$ may loosely be thought of as a smooth approximation to the $k$-variable {\sc And} function; recall that on input $(b_1,\dots,b_k) \in \bits^k$, the {\sc And} function outputs $-1$ iff $(b_1,\dots,b_k) = (-1,\dots,-1).$  (We note that \cite{HKM12} only require the $s=4$ case of the above theorem (this is their Theorem~3.5), since in their framework they can obtain perfect cancellation of the first, second and third derivative terms in the relevant difference of Taylor expansions.  In contrast we need to use all of the $s=1,2,3,4$ cases.)

As mentioned earlier, in our analysis of $\psi^\ast_{\lambda,k+1,\blue{(\vec{\theta},0)}}$ the \blue{last} argument will always receive a Boolean value from $\bits$ (corresponding to the output of the CNF $G$).  We will use the following simple claim to control the behavior of $\psi^\ast_{\lambda,k+1,\blue{(\vec{\theta},0)}}$ on inputs of this sort:

\begin{claim} \label{claim:psi-extra-coordinate}
Given $0<\lambda<1,k \geq 1$, and $\vec{\theta}\in \R^k$, let $v \in \R^k$ be such that $v \notin \strip_{\lambda,k, \vec{\theta}}$.  Then both vectors $\blue{(v,-1)} \in \R^{k+1}$ and
$\blue{(v,1)} \in \R^{k+1}$ lie outside of $\strip_{\lambda,k+1,\blue{(\vec{\theta},0)}}$.
\end{claim}
\begin{proof}
If $v \in \outter_{\lambda,k, \vec{\theta}}$ (because of some coordinate $v_j \geq \theta_j + \lambda$), then it is clear that $\blue{(v,1)}$ and $\blue{(v,-1)}$ both lie in $\outter_{\lambda,k+1,\blue{(\vec{\theta},0)}}$ (because of the same coordinate).  So suppose that $v \in \inner_{k,\vec{\theta}}.$  The vector $\blue{(v,1)}$ lies in $\outter_{\lambda,k+1,\blue{(\vec{\theta},0)}}$ (because of the \blue{last} coordinate $1>\lambda$), and the vector $\blue{(v,-1)}$ is easily seen to lie in $\inner_{k+1, \blue{(\vec{\theta},0)}}$.
\end{proof}

\subsection{Towards Theorem~\ref{thm:prg-informal}:  fooling the test function $\psi^\ast_{\lambda,k+1,\blue{(\vec{\theta},0)}}$}
As an intermediate step towards Theorem~\ref{thm:prg-informal} we will first establish the following ``pseudorandom generator'' for the smooth function $\psi^\ast_{\lambda,k+1,\blue{(\vec{\theta},0)}}$:

\begin{theorem}[$\Gen$ fools the smooth test function $\psi^\ast_{\lambda,k+1,\blue{(\vec{\theta},0)}}$]
\label{thm:our-smooth-prg} 
Let $\blue{H \wedge G}$ be a $(k,s,\tau)$-$\CNFLTF$, and let $W \in \R^{n\times k}$  be the matrix of weight vectors (each of norm 1) of the $\tau$-regular LTFs that comprise $H$, and $\vec{\theta}\in \R^k$ be the vector of their thresholds (so $\sign(W^j \cdot x - \theta_j)$ is the $j$-th LTF). For $0 <  \lambda < 1$, let $\psi^\ast_{\lambda,k+1,\blue{(\vec{\theta},0)}}: \R^{k+1}\to [-1,1]$ be as described in Fact~\ref{fact:cor-3.6}.  Then  when $\Gen$ is instantiated with the parameters from Section~\ref{sec:our-PRG},
\begin{align}
&\left| \Ex_{\bY\leftarrow\Gen}\big[ \psi^\ast_{\lambda,k+1,\blue{(\vec{\theta},0)}}(W^T \bY,\blue{G(\bY)} ) \big]  - \Ex_{\bU\leftarrow\bits^n}\big[ \psi^\ast_{\lambda,k+1,\blue{(\vec{\theta},0)}}(W^T \bU,\blue{G(\bU)}) \big]  \right| \nonumber \\
&=O\left(
{\frac {(\log k)^3}{\lambda^4}}  \left((\log k)^3 \cdot \tau \log(1/\tau) + {\frac 1 \tau} \cdot \delta_\CNF \cdot n^2\right) + {\frac 1 \tau}  \left(\sqrt{\delta_\CNF} +   \sum_{a=1}^3 n^a\sqrt{\delta_\CNF} \cdot \frac{(\log k)^{a-1}}{\lambda^a}\right)
\right).\label{eq:smooth-inv-cnfltf}
\end{align} 
\end{theorem}

\section{Setup for our coupling-based hybrid argument}
We begin by defining the sequence of random variables that we will use to hybridize between $\bY \leftarrow \Gen$, the $n$-bit pseudorandom input, and $\bU$, the $n$-bit uniform random input. 

\begin{definition}[Hybrid random variables] 
For any index $b \in \{0,1,\ldots,\ell\}$ and any hash $h \colon [n] \to [\ell]$, we define the hybrid random variable $\bX^{h,b}$ over $\bits^n$ as follows:  Independently across each $c \in [\ell]$,
\begin{itemize}

\item If $c > b$, then the coordinates $\bX^{h,b}_{h^{-1}(c)}$ of $\bX^{h,b}$ are distributed according to a uniform random draw from $\bits^{n}$;

\item If $c \leq b$, then the coordinates $\bX^{h,b}_{h^{-1}(c)}$ of $\bX^{h,b}$ are distributed according to a draw from an $r_\bucket$-wise independent random variable over $\bits^n$. 

\end{itemize}

Let $\calH$ be a $(2\log k)$-wise independent family of hashes $h : [n] \to [\ell]$.  For each $b \in \{0,1,\dots,\ell\}$, the hybrid random variable $\bX^{\bh,b}$ is defined by drawing $\bh \leftarrow \calH$ and then taking $\bX^{\bh,b}$ as above.  

\begin{remark}
\label{rem:first-and-last} 
 Note that $\bX^{\bh,0}$ is a uniform random variable over $\bits^n$ (indeed $\bX^{h,0}$ is uniform for every fixed hash $h$),
while $\bX^{\bh,\ell}$ is distributed according to $\Gen$. 
\end{remark} 
\end{definition}

\subsection{Coupling adjacent random variables in the hybrid argument} 

Fix a hash $h \colon [n] \to [\ell]$, a bucket $b \in [\ell]$, and a restriction $\rho \in \bits^{[n] \setminus h^{-1}(b)}$ fixing the variables outside bucket $h^{-1}(b)$. Recall that $\bX^{h,b-1}$ is distributed according to the uniform distribution within $h^{-1}(b)$, and $\bX^{h,b}$ is distributed according to a $r_\bucket$-wise independent distribution within this same bucket $h^{-1}(b)$.
For the remainder of this paper, for notational clarity unless otherwise indicated $\bU$ denotes a uniformly distributed random variable over $\bits^{h^{-1}(b)}$ and $\bZ$ denotes a $r_\bucket$-wise independent random variable over $\bits^{h^{-1}(b)}$.

\paragraph{Our CNF-fooling-based coupling.}

By the results of Bazzi and Razborov (Theorem~\ref{thm:BR}) and the choice of $r_\bucket$ from Section~\ref{sec:our-PRG}, the random variable  $\bZ$ $\delta_\CNF$-fools $G \uhr \rho$ (which, like $G$, is an $M$-clause CNF).  Consequently there exists a coupling $(\hat{\bU}, \hat{\bZ})$ between $\bU$ and $\bZ$ such that 
\begin{equation}
\Prx_{(\hat{\bU},\hat{\bZ})}\big[ (G\uhr \rho)(\hat{\bU}) \ne (G\uhr \rho)(\hat{\bZ}) \big] \le \delta_\CNF. \label{eq:coupling}
\end{equation} 
(Note that this coupling depends on $G\uhr \rho$.)

Consider the following joint distribution over a pair of random variables $(\hat{\bX}^{h,b-1}(\rho),\hat{\bX}^{h,b}(\rho))$, both supported on $\bn$:  First make a draw $(\hat{U},\hat{Y}) \leftarrow (\hat{\bU},\hat{\bZ})$,  and output $(\hat{X}^{h,b-1}(\rho),\hat{X}^{h,b}(\rho))$ where  
\begin{itemize}
\item $\hat{X}^{h,b-1}(\rho)$ assigns variables according to $\hat{U}$ within $h^{-1}(b)$, and according to $\rho$ outside $h^{-1}(b)$.
\item $\hat{X}^{h,b}(\rho)$ assigns variables according to $\hat{Y}$ within $h^{-1}(b)$, and according to $\rho$ outside $h^{-1}(b)$. 
\end{itemize} 

\begin{remark} 
\label{rem:average} 
Note that for $\brho \leftarrow \bX^{h,b}_{[n] \setminus h^{-1}(b)}$, we have that $\hat{\bX}^{h,b-1}(\brho)$ is distributed identically as $\bX^{h,b-1}$ and likewise $\hat{\bX}^{h,b}(\brho)$ is distributed identically as $\bX^{h,b}$. 
\end{remark}

\section{The hybrid argument: Proof of Theorem~\ref{thm:our-smooth-prg}}

Throughout this section for notational clarity we simply write $\psi$ instead of $\psi^\ast_{\lambda,k+1,\blue{(\vec{\theta},0)}}$. We also write $F_\psi : \bn \to [-1,1]$ to denote the function
\[ F_\psi(x) = \psi(W^T x,\blue{G(x)}). \] 
 
Our core technical result, which we prove in Section~\ref{sec:one-step}, is the following: 

\begin{lemma}[Error incurred in one step of hybrid] 
\label{lem:one-bucket-hybrid} 
For all hashes $h : [n] \to [\ell]$, buckets $b \in [\ell]$, and restrictions $\rho \in \bits^{[n]\setminus h^{-1}(b)}$, we have that 
\begin{align}
& \big| \E\big[ F_\psi(\hat{\bX}^{h,b-1}(\rho)) \big]-\E\big[ F_\psi(\hat{\bX}^{h,b}(\rho)) \big]\big| \label{eq:fixedhrho} \\
&= O\left( {\frac {(\log k)^3}{\lambda^4}}  \left((\log k)^2 \cdot h(W,b) + \delta_\CNF \cdot n^2 \right) + \sqrt{\delta_\CNF}+   \sum_{a=1}^3 n^a\sqrt{\delta_\CNF} \cdot \frac{(\log k)^{a-1}}{\lambda^a}\right), \nonumber 
\end{align} 
where 
\[ h(W,b) := \left(\blue{\sum_{j=1}^{k}} \| W^{j}_{h^{-1}(b)}\|^{4\log k} \right)^{1/\log k}. \]  
\end{lemma} 

The following corollary follows as an immediate consequence of Lemma~\ref{lem:one-bucket-hybrid}, Remark~\ref{rem:average}, and the triangle inequality:  

\begin{corollary}[Averaging Lemma~\ref{lem:one-bucket-hybrid} over $\brho$ and summing over {$b \in [\ell]$}] 
\label{cor:average-sum} 
For all hashes $h : [n] \to [\ell]$, we have that 
\begin{align*}
\big| \E\big[ F_\psi(\bX^{h,0}) \big]-\E\big[ F_\psi(\bX^{h,\ell}) \big]\big| &=  {\frac {O((\log k)^3)}{\lambda^4}} \cdot (\log k)^2 \cdot \sum_{b=1}^{\ell} h(W,b) \\
& \ + \ell \cdot O\left(  {\frac {(\log k)^3}{\lambda^4}} \cdot \delta_\CNF \cdot n^2  + \sqrt{\delta_\CNF} +   \sum_{a=1}^3 n^a\sqrt{\delta_\CNF} \cdot \frac{(\log k)^{a-1}}{\lambda^a}\right).
\end{align*} 
\end{corollary} 

\begin{proof} We have that 
\begin{align*} 
\big| \E\big[ F_\psi(\bX^{h,0}) \big]-\E\big[ F_\psi(\bX^{h,\ell}) \big]\big|  &\le \sum_{b=1}^{\ell} \big| \E\big[ F_\psi(\bX^{h,b-1}) \big]-\E\big[ F_\psi(\bX^{h,b}) \big]\big|   \tag*{(Triangle inequality)} \\
&= \sum_{b=1}^{\ell} \left| \Ex_{\brho\leftarrow\bX^{h,b}_{[n] \setminus h^{-1}(b)}}\big[ F_\psi(\hat{\bX}^{h,b-1}(\brho)) \big]-\Ex_{\brho\leftarrow\bX^{h,b}_{[n] \setminus h^{-1}(b)}}\big[ F_\psi(\hat{\bX}^{h,b}(\brho)) \big]\right|   \tag*{(Remark~\ref{rem:average})} \\ 
&\le  \sum_{b=1}^{\ell} \Ex_{\brho\leftarrow\bX^{h,b}_{[n] \setminus h^{-1}(b)}} \Big[ \big| \E\big[ F_\psi(\hat{\bX}^{h,b-1}(\brho)) \big]-\E\big[ F_\psi(\hat{\bX}^{h,b}(\brho)) \big]\big| \Big],
\end{align*} 
which gives the claimed bound via Lemma~\ref{lem:one-bucket-hybrid}. 
\end{proof} 

We do not have a good bound on the quantity $h(W,b)$ for an arbitrary hash $h : [n] \to [\ell]$ and bucket $b \in [\ell]$.  Instead, we shall use the following:

\begin{lemma}[Lemma 4.1 of~\cite{HKM12}]
\label{lem:HKM4.1} 
For $\ell = 1/\tau$ and $\calH$ a $(2\log k)$-wise independent hash family,  
\[ \Ex_{\bh\leftarrow \calH} \Bigg[ \sum_{b=1}^{\ell}  \bh(W,b)\Bigg] \le \sum_{b=1}^{\ell} \left(\Ex_{\bh\leftarrow\calH} \left[ \blue{\sum_{j=1}^{k}} \| W^j_{h^{-1}(b)}\|^{4\log k}\right]\right)^{1/\log k} \le  4\log k\cdot  \tau \log(1/\tau). \]
\end{lemma} 
\noindent (The middle quantity is what~\cite{HKM12} denotes by $\calH(W)$ and is the quantity they bound; the left inequality is by the power-mean inequality.)

We are now ready to prove Theorem~\ref{thm:our-smooth-prg}:

\begin{proof}[Proof of Theorem~\ref{thm:our-smooth-prg} assuming Lemma~\ref{lem:one-bucket-hybrid}] 
\begin{align*} 
& \left| \Ex_{\bY\leftarrow\Gen}\big[ \psi^\ast_{\lambda,k+1,\blue{(\vec{\theta},0)}}(W^T \bY, \blue{G(\bY)}) \big]  - \Ex_{\bU\leftarrow\bits^n}\big[ \psi^\ast_{\lambda,k+1,\blue{(\vec{\theta},0)}}(W^T \bU,\blue{G(\bU)}) \big]  \right| \\
     &= \big| \Ex\big[ F_\psi(\bX^{\bh,0})\big] - \Ex\big[ F_\psi(\bX^{\bh,\ell})\big] \big| \tag*{(Remark~\ref{rem:first-and-last} and definition of $F_\psi$)} \\
     &\le \Ex_{\bh\leftarrow\calH} \Big[ \big| \Ex\big[ F_\psi(\bX^{\bh,0})\big] - \Ex\big[ F_\psi(\bX^{\bh,\ell})\big] \big|\Big] \\
     &=  O\left(
{\frac {(\log k)^3}{\lambda^4}}  \left((\log k)^3 \cdot \tau \log(1/\tau) + {\frac 1 \tau} \cdot \delta_\CNF \cdot n^2\right) + {\frac 1 \tau}  \left(\sqrt{\delta_\CNF} +   \sum_{a=1}^3 n^a\sqrt{\delta_\CNF} \cdot \frac{(\log k)^{a-1}}{\lambda^a}\right)
\right),\end{align*} 
where the final equality is by Corollary~\ref{cor:average-sum}, Lemma~\ref{lem:HKM4.1}, and recalling that $\ell = 1/\tau$. 
\end{proof}

\section{A single step of the hybrid argument: Proof of Lemma~\ref{lem:one-bucket-hybrid}} 
\label{sec:one-step} 
Fix a hash $h : [n] \to [\ell]$, a bucket $b \in [\ell]$, and a restriction $\rho \in \bits^{[n]\setminus h^{-1}(b)}$.  
As is standard in applications of the Lindeberg method, we will express $F_\psi(\hat{\bX}^{h,b-1}(\rho))$ and $F_\psi(\hat{\bX}^{h,b}(\rho))$ as $\psi(\bv + \bDelta^\unif)$ and $\psi(\bv + \bDelta^\pseudo)$ respectively, {where $\bv$ is common to both random variables.  (Very roughly speaking, the Lindeberg method employs Taylor's theorem to show that quantities such as (\ref{eq:fixedhrho}) are small if $\bDelta^{\unif}$ and $\bDelta^{\pseudo}$ are sufficiently ``small" and $\psi$ is sufficiently ``nice.")}.   We now describe the choice of random variables $\bv, \bDelta^{\unif},\bDelta^{\pseudo}\in \R^{k+1}$ to accomplish this.  

We define $v : \bits^{h^{-1}(b)} \to \R^{k+1}$ as follows: 
\blue{
\begin{align*}
v(x)_j &= \sum_{i \in [n]\setminus h^{-1}(b)} W^{j}_i \rho_i  \qquad	\text{for $j \in [k]$}, \\
v(x)_{k+1} &= (G\uhr \rho)(x).  
\end{align*}
}
Recalling that $\rho$ is a fixed restriction, we observe that only the \blue{final} coordinate of $v$ depends on its input $x$. 
We further define $\Delta^{\unif} : \bits^{h^{-1}(b)} \to \R^{k+1}$ and $\Delta^{\pseudo} : \bits^{h^{-1}(b)} \times \bits^{h^{-1}(b)} \to \R^{k+1}$ as follows:
\blue{ 
\begin{align}
\Delta^{\unif}(x)_j &= \sum_{i \in h^{-1}(b)} W^j_i x_i  \qquad \text{for $j\in [k]$,} \nonumber \\
\Delta^{\unif}(x)_{k+1} &= 0,  \label{eq:always-zero}
\end{align}}
and 
\blue{
\begin{align*}
\Delta^{\pseudo}(x,z)_j &= \sum_{i \in h^{-1}(b)} W^j_i z_i \qquad \text{for $j\in [k]$,} \\
\Delta^{\pseudo}(x,z)_{k+1} &=  (G \uhr \rho)(z) - (G \uhr \rho)(x).
		  \end{align*} 
}
We observe that 
\begin{align*}
F_\psi(\hat{\bX}^{h,b-1}(\rho)) &\equiv \psi(v(\bU) + \Delta^{\unif}(\bU)) \\
F_\psi(\hat{\bX}^{h,b}(\rho)) &\equiv \psi(v(\hat{\bU}) + \Delta^{\pseudo}(\hat{\bU},\hat{\bZ})), 
\end{align*} 
and so the desired quantity (\ref{eq:fixedhrho}) of Lemma~\ref{lem:one-bucket-hybrid} that we wish to upper bound may be re-expressed as
\begin{align} 
(\ref{eq:fixedhrho}) &= 
 \big| \E[ F_\psi(\hat{\bX}^{h,b-1}(\rho))] - \E[F_\psi(\hat{\bX}^{h,b}(\rho)) ] \big|  \nonumber \\
&= \Big| \Ex_{\bU} \big[ \psi(v(\bU) + \Delta^{\unif}(\bU))\big] - \Ex_{(\hat{\bU},\hat{\bZ})}\big[ \psi(v(\hat{\bU}) + \Delta^{\pseudo}(\hat{\bU},\hat{\bZ})) \big]\Big|. \label{eq:our-goal}
\end{align} 

We observe that unlike standard Lindeberg-style proofs of invariance principles and associated pseudorandomness results, in our setup $v(\bU)$ and $\Delta^{\unif}(\bU)$ are not independent, and likewise neither are $v(\hat{\bU})$ and $\Delta^{\pseudo}(\hat{\bU},\hat{\bZ})$.  This motivates the definitions of the following subsection. 

\subsection{Mixtures of conditional distributions} 

Let $\bU^1$ denote the distribution $\bU$ conditioned on outcomes $x \in \bits^{h^{-1}(b)}$ such that $(G \uhr \rho)(x) = 1$, and similarly $\bU^{-1}$. Equivalently, $\bU^1$ and $\bU^{-1}$ are uniform distributions over $(G \uhr \rho)^{-1}(1)$ and $(G \uhr \rho)^{-1}(-1)$ respectively.  We note that $\bU$ can be expressed as the mixture of $\bU^1$ and $\bU^{-1}$ with mixing weights 
\begin{align*} 
\pi_1 &:= \Prx_{\bU}\big[ (G\uhr\rho)(\bU) = 1\big]  \\
\pi_{-1} &:= \Prx_{\bU}\big[ (G\uhr\rho)(\bU) = -1\big]. 
\end{align*}  
We may suppose without loss of generality that $\Prx_{\bU}[(G \uhr \rho)(\bU)=-1] \geq \Pr_{\bZ}[(G \uhr \rho)(\bZ)=-1]$ (the other case is entirely similar).     

Next, we similarly express the joint distribution $(\hat{\bU}, \hat{\bZ})$ as the mixture of conditional distributions $(\hat{\bU}^1, \hat{\bZ}^1)$, $(\hat{\bU}^{-1}, \hat{\bZ}^{-1})$, $(\hat{\bU}^{\err}, \hat{\bZ}^{\err})$, where 

\begin{itemize}
\item $(\hat{\bU}^1, \hat{\bZ}^1)$ is supported on pairs $(x,z)$ such that $(G \uhr \rho)(x)=(G \uhr \rho)(z)=1$
\item $(\hat{\bU}^{-1}, \hat{\bZ}^{-1})$ is supported on pairs $(x,z)$ such that $(G \uhr \rho)(x)=(G \uhr \rho)(z)=-1$
\item $(\hat{\bU}^{\err}, \hat{\bZ}^{\err})$ is supported on pairs $(x,z)$ such that $(G \uhr \rho)(x)=-1, (G \uhr \rho)(z)=1$.
\end{itemize} 
The mixing weights are $\tilde{\pi}_1, \tilde{\pi}_{-1}$, and $\tilde{\pi}_{\err}$ respectively, where 
\[ \tilde{\pi}_1 = \pi_1, \qquad \tilde{\pi}_{-1} = \pi_{-1} - \tilde{\pi}_{\err}, \qquad  \tilde{\pi}_\err \le \delta_\CNF\]
and the bound  $\tilde{\pi}_\err \leq \delta_\CNF$ follows from (\ref{eq:coupling}). 
We stress that while $\hat{\bU}^1$ is distributed identically as $\bU^1$, this is not the case for $\hat{\bU}^{-1}$ and $\bU^{-1}$, because of the small fraction of pairs that do not align perfectly under the coupling $(\hat{\bU},\hat{\bZ})$ and are captured by $(\hat{\bU}^\err,\hat{\bZ}^\err)$.

\begin{proposition}[Expressing $\bU$ and $(\hat{\bU},\hat{\bZ})$ as mixtures of conditional distributions]
\label{prop:mixture} 
For any function $f : \bits^{h^{-1}(b)} \to \R$, 
\[ \Ex_{\bU}\big[ f(\bU)\big] = \pi_1 \Ex_{\bU^1} \big[ f(\bU^1)\big]  + \pi_{-1} \Ex_{\bU^{-1}}\big[ f(\bU^{-1})\big]. \] 
Similarly, for any function $f : \bits^{h^{-1}(b)} \times \bits^{h^{-1}(b)} \to \R$, 
\begin{align*} 
& \Ex_{(\hat{\bU},\hat{\bZ})}\big[ f(\hat{\bU},\hat{\bZ})\big] \\
&= \tilde{\pi}_1 \Ex_{(\hat{\bU}^1,\hat{\bZ}^1)} \big[ f(\hat{\bU}^1,\hat{\bZ}^1) \big] + \tilde{\pi}_{-1}\Ex_{(\hat{\bU}^{-1},\hat{\bZ}^{-1})} \big[ f(\hat{\bU}^{-1},\hat{\bZ}^{-1}) \big] + \tilde{\pi}_{\err}\Ex_{(\hat{\bU}^{\err},\hat{\bZ}^{\err})}\big[ f(\hat{\bU}^{\err},\hat{\bZ}^{\err})\big] \\
&= \pi_1 \Ex_{(\hat{\bU}^1,\hat{\bZ}^1)} \big[ f(\hat{\bU}^1,\hat{\bZ}^1) \big] + (\pi_{-1} - \tilde{\pi}_{\err}) \Ex_{(\hat{\bU}^{-1},\hat{\bZ}^{-1})} \big[ f(\hat{\bU}^{-1},\hat{\bZ}^{-1}) \big] + \tilde{\pi}_{\err}\Ex_{(\hat{\bU}^{\err},\hat{\bZ}^{\err})}\big[ f(\hat{\bU}^{\err},\hat{\bZ}^{\err})\big] \\
&= \pi_1 \Ex_{(\hat{\bU}^1,\hat{\bZ}^1)} \big[ f(\hat{\bU}^1,\hat{\bZ}^1) \big] + \pi_{-1}  \Ex_{(\hat{\bU}^{-1},\hat{\bZ}^{-1})} \big[ f(\hat{\bU}^{-1},\hat{\bZ}^{-1}) \big] \pm 2 \,\delta_\CNF\cdot  \| f \|_\infty.
\end{align*} 
\end{proposition} 

These conditional distributions are useful because of the following two simple but crucial observations:

\begin{observation}[$\bv$ becomes constant] 
\label{obs:v} 
Fix $c \in \bits$.  For all $x \in \supp(\bU^c)$ we have that $v(x)$ is the same fixed vector $v^\ast\in \R^{k+1}$ given by
\blue{
\begin{align*}
v^*_j &= \sum_{i \in [n] \setminus {h^{-1}(b)}} W^j_i \rho_i \qquad \text{for $j \in [k]$,} \\
v^*_{k+1} &= (G \uhr \rho)(x) = c. 
\end{align*}
}
The same is true for $\hat{\bU}^c$: for all $x \in \supp(\hat{\bU}^c)$ we have $v(x) = v^*$.  
\end{observation} 

Note that as a consequence of Observation~\ref{obs:v}, the random variables $v(\bU^c)$ and $\Delta^{\unif}(\bU^c)$ are independent for $c \in \bits$, and likewise $v(\hat{\bU}^c)$ and $\Delta^{\pseudo}(\hat{\bU}^c,\hat{\bZ}^c)$ are independent as well; cf.~our remark following Equation (\ref{eq:our-goal}). The next observation further motivates our couplings $(\hat{\bU}^1,\hat{\bZ}^1)$ and $(\hat{\bU}^{-1},\hat{\bZ}^{-1})$:
\begin{observation}[$\Delta^{\pseudo}_{\blue{k+1}}= 0$] 
\label{obs:first} 
Fix $c \in \bits$.  For all $(\hat{U},\hat{Z}) \in \supp(\hat{\bU}^{c},\hat{\bZ}^{c})$, we have 
\[
\Delta^{\pseudo}_{\blue{k+1}}(\hat{U},\hat{Z}) = (G \uhr \rho)(\hat{Z}) - (G \uhr \rho)(\hat{U}) = 0. 
\]
\end{observation}

\subsubsection{Massaging our goal (\ref{eq:our-goal})} 

Applying Proposition~\ref{prop:mixture}, we can rewrite the RHS of (\ref{eq:our-goal}) as: 
\begin{align*}
&  \Big| \Ex_{\bU} \big[ \psi(v(\bU) + \Delta^{\unif}(\bU))\big] - \Ex_{(\hat{\bU},\hat{\bZ})}\big[ \psi(v(\hat{\bU}) + \Delta^{\pseudo}(\hat{\bU},\hat{\bZ})) \big]\Big| \\
& = \bigg| \Big(\pi_1 \Ex_{\bU^1} \big[ \psi(v(\bU^1) + \Delta^{\unif}(\bU^1))\big]
+ \pi_{-1} \Ex_{\bU^{-1}} \big[ \psi(v(\bU^{-1}) + \Delta^{\unif}(\bU^{-1}))\big]\Big) \\
&\ \ \ \ - \Big(\pi_1 \Ex_{(\hat{\bU}^1,\hat{\bZ}^1)} \big[ \psi(v(\hat{\bU}^1) + \Delta^{\pseudo}(\hat{\bU}^1,\hat{\bZ}^1)) \big] 
+ \pi_{-1} \Ex_{(\hat{\bU}^{-1}, \hat{\bZ}^{-1})} \big[ \psi(v(\hat{\bU}^{-1}) + \Delta^{\pseudo}(\hat{\bU}^{-1},\hat{\bZ}^{-1})) \big] \Big) \bigg| \\ 
&\ \ \ \ \pm 2\,\delta_\CNF \cdot \| \psi \|_\infty  \\ 
& \leq
 \pi_1 \cdot 
\bigg| \,  \Ex_{\bU^1} \big[ \psi(v(\bU^1) + \Delta^{\unif}(\bU^1))\big] -  \Ex_{(\hat{\bU}^1,\hat{\bZ}^1)} \big[ \psi(v(\hat{\bU}^1) + \Delta^{\pseudo}(\hat{\bU}^1,\hat{\bZ}^1)) \big]  \bigg| \\
&\ \ \ \ + \pi_{-1} \cdot \bigg| \Ex_{\bU^{-1}} \big[ \psi(v(\bU^{-1}) + \Delta^{\unif}(\bU^{-1}))\big] - \Ex_{(\hat{\bU}^{-1}, \hat{\bZ}^{-1})} \big[ \psi(v(\hat{\bU}^{-1}) + \Delta^{\pseudo}(\hat{\bU}^{-1},\hat{\bZ}^{-1})) \big] \big] \bigg| \\
&\ \ \ \ + 2\,\delta_\CNF, 
\end{align*}
where the final inequality uses the fact that $\psi$ has range $[-1,1]$.

We note that for $c \in \bits$, 
\[ \pi_c \cdot \bigg|\,  \Ex_{\bU^c} \big[ \psi(v(\bU^c) + \Delta^{\unif}(\bU^c))\big] -  \Ex_{(\hat{\bU}^c,\hat{\bZ}^c)} \big[ \psi(v(\hat{\bU}^c) + \Delta^{\pseudo}(\hat{\bU}^c,\hat{\bZ}^c)) \big] \bigg| \le 2\,\pi_c \cdot \| \psi \|_\infty = 2\,\pi_c,  \] 
which is at most $2 \sqrt{\delta_\CNF}$ if $\pi_c \leq \sqrt{\delta_\CNF}$ (this is the $O(\sqrt{\delta_\CNF})$ on the RHS of (\ref{eq:fixedhrho})).  We subsequently assume that $\pi_c \ge \sqrt{\delta_\CNF}$, and proceed to bound
\[
\sum_{c \in \bits} \pi_c \cdot \bigg|\,  \Ex_{\bU^c} \big[ \psi(v(\bU^c) + \Delta^{\unif}(\bU^c))\big] -  \Ex_{(\hat{\bU}^c,\hat{\bZ}^c)} \big[ \psi(v(\hat{\bU}^c) + \Delta^{\pseudo}(\hat{\bU}^c,\hat{\bZ}^c)) \big] \bigg|.
\]

\subsection{Applying Taylor's theorem} 

We proceed to analyze
\[
 \Ex_{\bU^c} \big[ \psi(v(\bU^c) + \Delta^{\unif}(\bU^c))\big] -  \Ex_{(\hat{\bU}^c,\hat{\bZ}^c)} \big[ \psi(v(\hat{\bU}^c) + \Delta^{\pseudo}(\hat{\bU}^c,\hat{\bZ}^c)) \big]  \]
for $c \in \{-1,1\}$.  We will do so by analyzing the Taylor expansion of $\psi(v + \Delta)$ (Fact~\ref{fact:taylor}):  

 \begin{align*}
 \psi(v + \Delta) &= \psi(v)  \tag*{(Zeroth-order term)} \\
&\ \ + \sum_{j \in [k+1]} (\partial_j \psi)(v) \Delta_j  \tag*{(First-order terms)}  \\
&\ \ + \sum_{j,j' \in [k+1]} \frac1{(j,j')!}  (\partial_{j,j'} \psi)(v) \Delta_j \Delta_{j'}   \tag*{(Second-order terms)} \\
&\ \ + \sum_{j,j',j'' \in [k+1]}\frac1{(j,j',j'')!}  (\partial_{j,j',j''} \psi)(v) \Delta_j \Delta_{j'} \Delta_{j''}   \tag*{(Third-order terms)} \\
&\ \ \pm \| \psi^{(4)} \|_1 \cdot \max_{j \in [k+1]} |\Delta_j|^4. \tag*{(Error term)} 
\end{align*}

Let us consider each of the five terms in the Taylor expansion, starting with the easiest one:

\begin{proposition} [Expected difference of zeroth-order terms]
\label{prop:zeroth} 
\[ \Ex_{\bU^c} \big[ \psi(v(\bU^c)) \big] - \Ex_{(\hat{\bU}^c,\hat{\bZ}^c)} \big[ \psi(v(\hat{\bU}^c)) \big] = 0. \] 
\end{proposition} 
\begin{proof}
Recalling Observation~\ref{obs:v}, we have that 
\[ v(x) = v(x') = v^* \]
for all $x \in \supp(\bU^c)$ and $x' \in \supp(\hat{\bU}^c)$, where $v^*$ is a fixed vector in $\R^{k+1}$.  In order words, the random variables $v(\bU^c)$  and $v(\hat{\bU}^c)$ are both supported entirely on the same constant $v^*$.
\end{proof} 

\subsubsection{Expected difference of third-order terms} 

In this section we bound the expected difference of the third-order terms: 
\begin{align} 
& \pi_c \cdot \bigg| \Ex_{\bU^c} \bigg[ \sum_{j,j',j'' \in [k+1]} (\partial_{j,j',j''} \psi)(v(\bU^c)) \Delta^{\unif}(\bU^c)_j \Delta^{\unif}(\bU^c)_{j'} \Delta^{\unif}(\bU^c)_{j''}\bigg] \nonumber \\
&\ \ \ \ \ \ -  \Ex_{(\hat{\bU}^c,\hat{\bZ}^c)} \bigg[ \sum_{j,j',j'' \in [k+1]} (\partial_{j,j',j''} \psi)(v(\hat{\bU}^c)) \Delta^{\pseudo}(\hat{\bU}^c,\hat{\bZ}^c)_j \Delta^{\pseudo}(\hat{\bU}^c,\hat{\bZ}^c)_{j'} \Delta^{\pseudo}(\hat{\bU}^c,\hat{\bZ}^c)_{j''}\bigg] \bigg|. \label{eq:third-order-a} 
\end{align}
We observe that in standard applications of the Lindeberg method the quantity analogous to the above quantity would be exactly zero due to matching moments (see the parenthetical following Equation~(\ref{eq:loe}) \violet{below}). Since our setting requires that we perform the hybrid argument over  the conditional distributions $\bU^c$ and $(\hat{\bU}^c,\hat{\bZ}^c)$ (rather than the global distributions $\bU$ and $(\hat{\bU},\hat{\bZ})$) we no longer have matching moments, but our analysis in this section shows that the error incurred by the mismatch is acceptably small.  More precisely, we will prove that (\ref{eq:third-order-a}) is at most $O(n^3 \sqrt{\delta_\CNF}\cdot (\log k)^2/\lambda^3)$.   An identical argument shows that the analogous quantities for the first- and second-order terms are at most $O(n^2\sqrt{\delta_\CNF}\cdot (\log k)/\lambda^2)$ and $O(n\sqrt{\delta_\CNF}/\lambda)$ respectively. 

We begin by noting that 
\begin{align}
(\ref{eq:third-order-a}) &= \pi_c\cdot  \Bigg| \sum_{j,j',j'' \in [k+1]} (\partial_{j,j',j''} \psi)(v^\ast)  \nonumber \\ 
& \ \ \ \ \ \ \ \ \ \  \ \ \ \ \ \ \ \ \ \ \ \ \   \underbrace{\bigg(\Ex_{\bU^c} \bigg[\prod_{\xi \in \{j,j',j''\}} \Delta^{\unif}(\bU^c)_\xi \bigg] - \Ex_{(\hat{\bU}^c,\hat{\bZ}^c)}\bigg[\prod_{\xi\in \{j,j',j''\}}\Delta^{\pseudo}(\hat{\bU}^c,\hat{\bZ}^c)_\xi\bigg] \bigg)}_{\Phi(j,j',j'')}  \Bigg| \label{eq:third-order} 
\end{align} 
where (as in Proposition~\ref{prop:zeroth}) we have again used Observation~\ref{obs:v} to get that $v(\bU^c) \equiv v(\hat{\bU}^c) \equiv v^*$ for a fixed vector $v^* \in \R^{k+1}$.

\begin{observation}[Difference is zero if $j=\blue{k+1}$ participates] 
\label{obs:only-j-ge-2} 
If $\blue{k+1} \in \{j,j',j''\}$ then $\Phi(j,j',j'') =0$. 
\end{observation}

\begin{proof} 
This is because $\Delta^{\unif}_\blue{k+1}$ is the identically 0 function (by definition; recall Equation (\ref{eq:always-zero})), and $\Delta_{\blue{k+1}}^{\pseudo}(\hat{U}^c,\hat{Z}^c) = 0$ for all $(\hat{U}^c,\hat{Z}^c) \in \supp(\hat{\bU}^c,\hat{\bZ}^c)$ (Observation~\ref{obs:first}). 
\end{proof} 

Therefore it suffices to reason about $\Phi(j,j',j'')$ for triples $j,j',j''\in \blue{[k]}$.   Fix any such triple.  Recalling the definitions of $\Delta^{\unif}_j$ and $\Delta^{\pseudo}_j$ for $j\in\blue{[k]}$: 
\[ \Delta^{\unif}(x)_j = \sum_{i \in h^{-1}(b)} W^j_i x_i,\qquad \Delta^{\pseudo}(x,z)_j = \sum_{i \in h^{-1}(b)} W^j_i z_i \] 
 and applying linearity of expectation, we have that 
\begin{equation} 
\Phi(j,j',j'') = \sum_{i,i',i''\in h^{-1}(b)} W^j_i W^{j'}_{i'} W^{j''}_{i''} \Big( \Ex_{\bU^c}\big[\bU^c_i\,\bU^c_{i'}\, \bU^c_{i''} \big] - \Ex_{(\hat{\bU}^c,\hat{\bZ}^c)}\big[\hat{\bZ}^c_i\,\hat{\bZ}^c_{i'}\, \hat{\bZ}^c_{i''} \big]\Big). \label{eq:loe} 
\end{equation} 
(Note that $\E\big[\bU_i \bU_{i'}\bU_{i''}\big] - \E\big[\hat{\bZ}_i \hat{\bZ}_{i'}\hat{\bZ}_{i''}\big] = 0$ since $\bU$ and $\bZ$ have matching moments. However, since we are working with the conditional distributions $\hat{\bU}^c$ and $(\hat{\bU}^c,\hat{\bZ}^c)$ this is no longer the case; nevertheless, we will now show that this difference is adequately small.)  The first expectation on the RHS can be expressed as $2p_\unif-1$ where 
\begin{equation}
p_\unif =  \Prx_{\bU^c}\big[\bU^c_i\,\bU^c_{i'}\,\bU^c_{i''} = 1\big] 
= \frac{\Prx_{\bU} \big[ \,\bU_i\,\bU_{i'}\,\bU_{i''} = 1, (G \uhr \rho)(\bU) = c\, \big]}{\Prx_{\bU} \big[\, (G \uhr \rho)(\bU) = c\, \big]},  \label{eq:punif} 
\end{equation}
and likewise the second expectation can be expressed as $2p_\pseudo -1$ where 
\begin{align} 
p_\pseudo &=  \Prx_{(\hat{\bU}^c,\hat{\bZ}^c)}\big[\hat{\bZ}^c_i\,\hat{\bZ}^c_{i'}\, \hat{\bZ}^c_{i''}  = 1\big] \nonumber \\ 
&= \frac{\Prx_{(\hat{\bU},\hat{\bZ})}\big[\, \hat{\bZ}_i\,\hat{\bZ}_{i'}\,\hat{\bZ}_{i''} = 1, (G \uhr\rho)(\hat{\bU}) = (G \uhr \rho)(\hat{\bZ})  = c\, \big]}{\Prx_{(\hat{\bU},\hat{\bZ})}\big[\,(G \uhr\rho)(\hat{\bU}) = (G \uhr \rho)(\hat{\bZ})  = c\, \big]}. \label{eq:ppseudo} 
\end{align} 
Note that the numerator of (\ref{eq:ppseudo}) is 
\begin{align*} 
& \Prx_{\bZ}\big[\, \bZ_i\,\bZ_{i'}\,\bZ_{i''} = 1,  (G \uhr \rho)(\bZ)  = c\, \big] - \Prx_{(\hat{\bU},\hat{\bZ})}\big[\, \hat{\bZ}_i,\hat{\bZ}_{i'}\,\hat{\bZ}_{i''} = 1, (G \uhr\rho)(\hat{\bU}) = -c, (G \uhr \rho)(\hat{\bZ})= c \,\big] \\
&\ge \Prx_{\bZ}\big[ \,\bZ_i\,\bZ_{i'}\,\bZ_{i''} = 1,  (G \uhr \rho)(\bZ)  = c \,\big] - \Prx_{(\hat{\bU},\hat{\bZ})}\big[\,(G \uhr\rho)(\hat{\bU}) = -c, (G \uhr \rho)(\hat{\bZ})= c\, \big] \\
&= \Prx_{\bZ}\big[ \,\bZ_i\,\bZ_{i'}\,\bZ_{i''} = 1,  (G \uhr \rho)(\bZ)  = c \,\big] - O(\delta_\CNF). \tag*{(by (\ref{eq:coupling}))}
\end{align*} 
Likewise, the denominator of (\ref{eq:ppseudo}) is $\Prx_{\bZ}\big[\,(G\uhr \rho)(\bZ)= c\,\big] - O(\delta_\CNF)$, again by (\ref{eq:coupling}).  Therefore, we have that 
\begin{align*} 
p_\pseudo  &= \frac{\Prx_{\bZ}\big[ \,\bZ_i\,\bZ_{i'}\,\bZ_{i''} = 1,  (G \uhr \rho)(\bZ)  = c\, \big] - O(\delta_\CNF)}{\Prx_{\bZ}\big[\,(G\uhr \rho)(\bZ)= c\,\big] - O(\delta_\CNF)}. 
\end{align*} 

Next, we note that $\bZ$ $\delta_\CNF$-fools the function $(G \uhr \rho)(x) \oplus \beta$ as well as the function $((G \uhr \rho)(x) \oplus \beta)  \wedge (\violet{\neg\,}(x_i \oplus x_{i'} \oplus x_{i''}))$ for $i,i',i'' \in h^{-1}(b)$, $\beta \in \bits$.  The former is true by Theorem~\ref{thm:BR} and the fact that $r_\bucket \geq O((\log (M/\delta_\CNF))^2)$, and the latter is true because $r_\bucket \geq 4 \log k + O((\log (M/\delta_\CNF))^2 \geq 3 + O((\log (M/\delta_\CNF))^2).$  (Observe that if a function $f(x)$ and all its restrictions are $\kappa$-fooled by $r$-wise independence, then $f(x) \wedge J(x)$, where $J$ is any 3-junta, is $\kappa$-fooled by $(r+3)$-wise independence.) Hence we have 
\begin{align*} 
p_\pseudo  
&=  \frac{\Prx_{\bU}\big[\, \bU_i\,\bU_{i'}\,\bU_{i''} = 1,  (G \uhr \rho)(\bU)  = c\, \big] \pm O(\delta_\CNF)}{\Prx_{\bU}\big[\,(G\uhr \rho)(\bU)= c\,\big] \pm O(\delta_\CNF)}.
\end{align*}

Since by assumption $\pi_c = \Prx_{\bU}\big[(G\uhr \rho)(\bU)= c\big] \ge \sqrt{\delta_\CNF}$, it follows from the above and (\ref{eq:punif}) that 
\[ p_\pseudo = p_\unif \pm O(\sqrt{\delta_\CNF}). \] 
Recalling (\ref{eq:loe}), we have shown that 
\[ |\Phi(j,j',j'')| =  \sum_{i,i',i''\in h^{-1}(b)} W^j_i W^{j'}_{i'} W^{j''}_{i''}  \cdot O(\sqrt{\delta_\CNF}) = O(n^3 \sqrt{\delta_\CNF}), 
\]
\violet{where}
the final equality uses the trivial bounds of $|W^j_i| \le 1$ for all $\blue{j \in [k]}$ and $i \in h^{-1}(b)$, and $|h^{-1}(b)| \le n$.   We conclude that the expected difference \violet{of} the third-order terms is at most 
\begin{align*} (\ref{eq:third-order}) &= \pi_c \cdot \bigg| \sum_{j,j',j'' \in [k+1]} (\partial_{j,j',j''} \psi)(v^*) \cdot \Phi(j,j',j'')  \bigg|  \\
&= O(n^3\sqrt{\delta_\CNF}) \cdot \bigg| \sum_{j,j',j'' \in [k+1]} (\partial_{j,j',j''} \psi)(v^*) \bigg|  \\
&= O(n^3\sqrt{\delta_\CNF}) \cdot \frac{(\log k)^2}{\lambda^3},
\end{align*} 
where the final equality uses the bound on $\| \psi^{(3)} \|_1$ given by Fact~\ref{fact:cor-3.6}. 
\subsubsection{Error term}  

Finally we bound the contribution from the error terms.  This is at most
\begin{align*}
& \sum_{c \in \{-1,1\}} \bigg(\pi_c \Ex_{\bU^c}\Big[ \| \psi^{(4)} \|_1  \max_{j \in [k+1]} \big|\Delta^{\unif}(\bU^c)_j\big|^4 \Big] + \pi_c \Ex_{(\hat{\bU}^c,\hat{\bZ}^c)}\Big[ \| \psi^{(4)} \|_1  \max_{j \in [k+1]} \big|\Delta^{\pseudo}(\hat{\bU}^c,\hat{\bZ}^c)_j\big|^4 \Big]\bigg) \\
&= \| \psi^{(4)} \|_1  \sum_{c \in \{-1,1\}} \bigg(\pi_c \Ex_{\bU^c}\Big[  \max_{\blue{j\in [k]}} \big|\Delta^{\unif}(\bU^c)_j\big|^4 \Big] + \pi_c \Ex_{(\hat{\bU}^c,\hat{\bZ}^c)}\Big[ \max_{\blue{j\in [k]}} \big|\Delta^{\pseudo}(\hat{\bU}^c,\hat{\bZ}^c)_j\big|^4 \Big]\bigg), 
\end{align*} 
where this equality again uses the fact that $\Delta^{\unif}_{\blue{k+1}}$ is the constant $0$ function (by definition; recall Equation (\ref{eq:always-zero})) and $\Delta^{\pseudo}(\hat{U}^c,\hat{Z}^c)_{\blue{k+1}} = 0$ for all $(\hat{U}^c,\hat{Z}^c) \in \supp(\hat{\bU}^c,\hat{\bZ}^c)$ (Observation~\ref{obs:first}) to get that the max's can be taken over $\blue{j \in [k]}$ rather than  $j \in [k+1]$.  Applying both statements of Proposition~\ref{prop:mixture}, we get that the above is 
\begin{align*} & \| \psi^{(4)} \|_1  \bigg( \Ex_{\bU} \Big[  \max_{\blue{j\in [k]}} \big|\Delta^{\unif}(\bU)_j\big|^4 \Big] + \Ex_{(\hat{\bU},\hat{\bZ})}\Big[ \max_{\blue{j\in [k]}} \big|\Delta^{\pseudo}(\hat{\bU},\hat{\bZ})_j\big|^4 \Big] \pm 2\,\delta_\CNF \cdot \Big\| \max_{\blue{j\in [k]}} \big|\Delta^{\pseudo}(\cdot ,\cdot )_j\big|^4 \Big\|_\infty \bigg) \\
&= {\frac {O((\log k)^3)}{\lambda^4}}  \bigg( \Ex_{\bU} \Big[  \max_{\blue{j\in [k]}} \big|\Delta^{\unif}(\bU)_j\big|^4 \Big] + \Ex_{(\hat{\bU},\hat{\bZ})}\Big[ \max_{\blue{j\in [k]}} \big|\Delta^{\pseudo}(\hat{\bU},\hat{\bZ})_j\big|^4 \Big] + \delta_\CNF (\sqrt{n})^4 \bigg),
\end{align*} 
where we have used Fact~\ref{fact:cor-3.6} and the easy bound $\|\Delta^\pseudo_j\|_\infty \leq \sqrt{n}$ for $\blue{j\in [k]}$ (recalling that each weight vector $W^j$ has $\| W^j\|_2$ equal to 1). Since $r_\bucket \geq 4 \log k$, by the same hypercontractivity-based calculations as in the proof of Claim~4.4 of \cite{HKM12} (starting at the bottom of page 15), each of the two expectations is at most $O((\log k)^2) \cdot h(W,b)$. (We refer the reader to Section 6.2 of \cite{HKM12} for a justification of why the $r_\bucket$-wise independence of the distribution $\hat{\bZ}$ suffices for the analysis of the second expectation.) This concludes the proof of Lemma~\ref{lem:one-bucket-hybrid}.

\section{Proving Theorem~\ref{thm:prg-informal} using Theorem~\ref{thm:our-smooth-prg}} \label{sec:putting-together}

In this section we relate what we have shown so far, a bound on 
\begin{equation} \label{eq:useme}
\Big|\Ex_{\bU\leftarrow \bn}\big[ F_{\psi^\ast_{\lambda,k+1,\blue{(\vec{\theta},0)}}}(\bU)\big] - \Ex_{\bY\leftarrow\Gen} \big[F_{\psi^\ast_{\lambda,k+1,\blue{(\vec{\theta},0)}}}(\bY)\big] \Big|, 
\end{equation} 
to the relevant quantity for Theorem~\ref{thm:prg-informal}, 
\begin{equation} \Big|\Ex_{\bU\leftarrow \bn}\big[ F(\bU)\big] - \Ex_{\bY\leftarrow\Gen} \big[F(\bY)\big] \Big|. \label{eq:final-error} 
\end{equation} 
By~\cite{HKM12}'s Lemma 3.3, the quantity (\ref{eq:final-error}) is at most 
\[
O(1) \cdot ((\ref{eq:useme}) + \Prx\big[ (W^T\bU,\blue{G(\bU)}) \in \strip_{\lambda,k+1,\blue{(\vec{\theta},0)}} \big]).
\]
We bound this probability as follows:
\begin{align*} 
 \Prx\big[ \,(W^T\bU,\blue{G(\bU)}) \in \strip_{\lambda,k+1,\blue{(\vec{\theta},0)}}\, \big] &\le \Prx\big[ \,W^T\bU \in \strip_{\lambda,k,\vec{\theta}} \,\big]  \tag*{(Claim~\ref{claim:psi-extra-coordinate})}  \\
 &\le \Prx\big[ \,W^T\bG \in \strip_{\lambda,k,\vec{\theta}} \,\big] + O((\log k)^{8/5} (\tau \log(1/\tau))^{1/5}) \tag*{(\cite{HKM12}'s invariance principle, Theorem~\ref{thm:HKM-invariance})} \\ 
 &= O(\lambda\sqrt{\log k}) + O((\log k)^{8/5} (\tau \log(1/\tau))^{1/5}). \tag*{(Theorem~\ref{thm:Gaussian-anticoncentration})}
\end{align*} 
Therefore, it follows that 
\begin{align*}
 (\ref{eq:final-error}) &= O\left(
{\frac {(\log k)^3}{\lambda^4}}  \left((\log k)^3 \cdot \tau \log(1/\tau) + {\frac 1 \tau} \cdot \delta_\CNF \cdot n^2\right) + {\frac 1 \tau}  \left(\sqrt{\delta_\CNF} +   \sum_{a=1}^3 n^a\sqrt{\delta_\CNF} \cdot \frac{(\log k)^{a-1}}{\lambda^a}\right)
\right) \\
 & \ \ \ + O(\lambda\sqrt{\log k}) + O((\log k)^{8/5} (\tau \log(1/\tau))^{1/5}).  
 \end{align*} 
 As in \cite{HKM12}, we choose $\lambda = (\log k)^{11/10} (\tau \log(1/\tau))^{1/5}$, which makes 
 \[
 \lambda \sqrt{\log k} = \Theta\left( {\frac {(\log k)^3}{\lambda^4}} \cdot \left((\log k)^3 \cdot \tau \log(1/\tau) \right) \right)= \Theta((\log k)^{8/5} (\tau \log(1/\tau))^{1/5})).\]
Since $k \leq 2^n$ and $\tau \geq 1/\sqrt{n}$, a suitable choice of $\delta_\CNF = 1/\poly(n)$ makes the remaining quantity,
\[
{\frac 1 \tau}  \left( {\frac {(\log k)^3}{\lambda^4}} \cdot \delta_\CNF \cdot n^2 + \sqrt{\delta_\CNF} +   \sum_{a=1}^3 n^a\sqrt{\delta_\CNF} \cdot \frac{(\log k)^{a-1}}{\lambda^a}\right),
\]
at most $O((\log k)^{8/5} (\tau \log(1/\tau))^{1/5})),$ so we get that $ (\ref{eq:final-error})$ is $O((\log k)^{8/5} (\tau \log(1/\tau))^{1/5})$ as desired.  This concludes the proof of Theorem~\ref{thm:prg-informal}.
 \qed

\bibliography{allrefs}{}
\bibliographystyle{alpha}

\end{document}